\documentclass{article}

\usepackage{hyperref}
\usepackage{amssymb,amsmath,xspace,enumerate,gensymb}
\usepackage{amsthm}

\usepackage[pagewise]{lineno}
\usepackage{xspace}
\usepackage{arydshln}
\usepackage{rotating}
\usepackage{listings}
\usepackage{soul}
\usepackage{ifmtarg}
\usepackage{stmaryrd} \usepackage{sidecap}
\usepackage{cleveref}[2012/02/15]

\usepackage{tikz}
\usetikzlibrary{arrows,decorations.pathmorphing,backgrounds,positioning,fit,petri,backgrounds}
\usetikzlibrary{calc}
\usetikzlibrary{scopes}
\usetikzlibrary{patterns}
\usepgflibrary{shapes.geometric}
\usetikzlibrary{matrix}

\usepackage{pgffor}
\usepackage{paralist}
\usepackage{chngcntr}

\newif\ifcomments
\newif\ifchanges
\commentsfalse\changesfalse
\commentstrue\changestrue

  \makeatletter{}\newcommand{\mtext}[1]{\textsc{#1}}

\newcommand{\ramsey}[2]{\ensuremath{R_{#1}(#2)}}
\newcommand{\homramsey}[2]{\ensuremath{R^{\text{hom}}_{#1}(#2)}}

\newcommand{\subseq}{\sqsubseteq}

\newcommand{\mand}{\wedge}
\newcommand{\mor}{\vee}

\newcommand{\restrict}[2]{#1\mspace{-3mu}\upharpoonright \mspace{-3mu}#2}
\newcommand{\isomorph}{\simeq}
\newcommand{\isomorphVia}[1]{\isomorph_{#1}}
\newcommand{\swap}[2]{id{[#1, #2]}}

\newcommand{\sem}[2]{\ensuremath{\llbracket #1\rrbracket_{#2}}}

\newcommand{\nd}{d}

\newcommand{\norder}{\prec}

\newcommand{\col}{col}

\newcommand{\N}{\ensuremath{\mathbb{N}}}

\newcommand  {\myclass} [1]  {\ensuremath{\textsc{#1}}}

\newcommand{\mneg}{\neg} 
\newcommand{\class}{\calC}

\newcommand{\DynClass}[1]{\myclass{Dyn#1}\xspace}

\newcommand{\DynProp}{\DynClass{Prop}}
\newcommand{\DynPropbi}{\ensuremath{\DynClass{Prop}^*}\xspace}

\newcommand{\DynQF}{\DynClass{QF}}
\newcommand{\DynQFbi}{\ensuremath{\DynClass{QF}^*}\xspace}
\newcommand{\DynFO}{\DynClass{FO}}

\newcommand{\DynC}{\DynClass{$\class$}}

\newcommand{\DynAND}{\DynClass{PropCQ}}
\newcommand{\DynAnd}{\DynAND}
\newcommand{\DynPropCQ}{\DynAND}

\newcommand{\DynPropPos}{\DynClass{PropUCQ}}

\newcommand{\DynPropUCQ}{\DynPropPos}

\newcommand{\DynAndNeg}{\DynClass{PropCQ{\ensuremath{^{\mneg}}}}}
\newcommand{\DynPropCQneg}{\DynAndNeg}

\newcommand{\updates}{\ensuremath{\Delta}}

\newcommand{\abstrUpd}{\ensuremath{\updates}}

\newcommand{\init}{\mtext{Init}}

\newcommand{\First}{\mtext{First}}
\newcommand{\List}{\mtext{List}}
\newcommand{\Last}{\mtext{Last}}
\newcommand{\In}{\mtext{In}}
\newcommand{\Out}{\mtext{Out}}
\newcommand{\Empty}{\mtext{Empty}}
\newcommand{\Succ}{\mtext{Succ}}
\newcommand{\Pred}{\mtext{Pred}}

\newcommand{\ins}{\mtext{ins}}
\newcommand{\del}{\mtext{del}}

\newcommand{\type}[2]{\ensuremath{\langle #1, #2 \rangle}}
\newcommand{\stype}[3]{\ensuremath{\langle #1, #2 \rangle_{#3}}}

\newcommand{\insertdescr}[2]{\textbf{Insertion of \ensuremath{#2} into \ensuremath{#1}.}}
\newcommand{\deletedescr}[2]{\textbf{Deletion of \ensuremath{#2} from \ensuremath{#1}.}}

\newcommand{\schema}{\tau}
\newcommand{\relSchema}{\schema_{\text{rel}}}
\newcommand{\conSchema}{\schema_{\text{const}}}
\newcommand{\funSchema}{\schema_{\text{fun}}}
\newcommand{\Terms}[2]{\textsc{Terms}^{#2}_{#1}} 
\newcommand{\inpSchema}{\schema_{\text{in}}}
\newcommand{\auxSchema}{\schema_{\text{aux}}}

\newcommand{\builtinSchema}{\schema_{\text{bi}}}
\newcommand{\arity}{\ensuremath{\text{Ar}}}

\newcommand{\auxInit}{\init_{\text{aux}}}
\newcommand{\builtinInit}{\init_{\text{bi}}}

\newcommand{\nb}[3]{\calN_{#2}^{#3}(#1)}
\newcommand{\nbv}[3]{\vec \calN_{#2}^{#3}(#1)}

\newcommand{\struc}{\calS}
\newcommand{\db}{\calD}
\newcommand{\inp}{\calI}
\newcommand{\aux}{\calA}
\newcommand{\builtin}{\calB}
\newcommand{\domain}{D}
\newcommand{\query}{\calQ}
\newcommand{\querys}{Q}

\newcommand{\state}{\ensuremath{\struc}}

\newcommand{\prog}{\ensuremath{\calP}\xspace}

\newcommand{\updateDB}[2]{\ensuremath{#1(#2)}}
\newcommand{\updateState}[3]{\ensuremath{#1_{#2}(#3)}}

\makeatletter \newcommand{\uf}[4]{
  \@ifmtarg{#4}{
    \ensuremath{\phi^{#1}_{#2}(#3)}
   }{
    \ensuremath{\phi^{#1}_{#2}(#3; #4)}
  }
}
\newcommand{\ufwa}[2]{
  \ensuremath{\phi^{#1}_{#2}}
}

\newcommand{\ut}[4]{\ensuremath{t^{#1}_{#2}(#3; #4)}}

\newcommand{\ite}[3]{\mtext{ite}(#1,#2,#3)}

\newcommand {\calA}      {{\mathcal A}\xspace}
\newcommand {\calB}      {{\mathcal B}\xspace}
\newcommand {\calC}      {{\mathcal C}\xspace}
\newcommand {\calD}      {{\mathcal D}\xspace}

\newcommand {\calI}      {{\mathcal I}\xspace}

\newcommand {\calN}      {{\mathcal N}\xspace}

\newcommand {\calP}      {{\mathcal P}\xspace}
\newcommand {\calQ}      {{\mathcal Q}\xspace}

\newcommand {\calS}      {{\mathcal S}\xspace}
\newcommand {\calT}      {{\mathcal T}\xspace}

\newcommand  {\problemdescr} [3] {
    \vspace{3mm}
    \def\Name{#1}
    \def\Input{#2}
    \def\Question{#3}
	  \hspace{5mm}\begin{tabular}{r l}	    \textit{Query:} & \textsc{\Name} \\
	    \textit{Input:} & \Input \\
	    \textit{Question:} & \Question
	  \end{tabular}
    \vspace{3mm}
    }

\newcommand {\problem} [1] {\textsc{#1}}

\newcommand{\dynProb}[1]{\textsc{Dyn(#1)}}

\newcommand{\stgraph}{$s$-$t$-graph\xspace}
\newcommand{\stgraphs}{$s$-$t$-graphs\xspace}

\newcommand{\reachQ}{\textsc{Reach}\xspace}\newcommand{\streachQ}{\textsc{$s$-$t$-Reach}\xspace}\newcommand{\streachabilityquery}{$s$-$t$-reachability query\xspace}
\newcommand{\dynstReachQ}{\dynProb{\textsc{$s$-$t$-Reach}}\xspace}
\newcommand{\stTwoPath}{\problem{$s$-$t$-Two\-Path}\xspace}
\newcommand{\dynstTwoPath}{\dynProb{\stTwoPath}\xspace}
\newcommand{\sTwoPath}{\problem{$s$-Two\-Path}\xspace}
\newcommand{\dynsTwoPath}{\dynProb{\sTwoPath}\xspace}

\newcommand{\clique}[1]{\problem{$#1$-Clique}\xspace}
\newcommand{\dynClique}[1]{\dynProb{$#1$-Clique}\xspace}

\newcommand{\colorability}[1]{\problem{$#1$-Col}\xspace}
\newcommand{\dynColorability}[1]{\dynProb{$#1$-Col}\xspace}

\newcommand{\streach}{$s$-$t$-Reach}
\newcommand{\streachp}{\problem{\streach}\xspace}

\makeatletter 
     \newtheorem{theorem}{Theorem}[section]
   \newtheorem{lemma}[theorem]{Lemma}
   \newtheorem{corollary}[theorem]{Corollary}
   \newtheorem{proposition}[theorem]{Proposition}
   
    \theoremstyle{definition}
   \newtheorem{definition}{Definition}
   \newtheorem{example}{Example}

   \newenvironment{proofsketch}{\noindent\emph{Proof sketch.}\enspace}{\qed \medskip}
   \newenvironment{proofof}[1]{\noindent\textsc{Proof (of #1).}\enspace}{\qed}
  \newenvironment{proofsketchof}[1]{\noindent\textsc{Proof sketch (of #1).}\enspace}{\qed}

\makeatletter
\newcommand{\theoremcont}[3]{
   \def\Type{#1}
   \def\Number{#2}
   \def\Label{#3}
  \@ifmtarg{#3}{
    \vspace{2mm}\noindent\textbf{\Type\ \Number.}\itshape
   }{
    \vspace{2mm}\noindent\textbf{\Type\ \Number}\ \itshape(\Label).
  }
}

\newcommand{\df}{\ensuremath{\mathrel{\smash{\stackrel{\scriptscriptstyle{
    \text{def}}}{=}}}}}

\providecommand{\nc}{\newcommand}

\ifcomments
\nc{\commentbox}[1]{\noindent\framebox{\parbox{\linewidth}{#1}}}
\nc{\todo}[1]{\ \\ {\color{red} \fbox{\parbox{\linewidth}{{\sc
          ToDo}:\\  #1}}}}

\newcommand{\acomment}[2]{\ \\ \fbox{\parbox{\linewidth}{{\sc #1}:\\ #2}}}
\newcommand{\mcomment}[2]{{\color{blue}(#1)}\footnote{#1: #2}}                                 \else
\nc{\commentbox}[1]{}
\newcommand{\mcomment}[2]{}
\newcommand{\acomment}[2]{}
\fi

\ifchanges

\newcommand{\loldnew}[3]{\commentbox{{\textcolor{blue}{\setlength{\fboxsep}{1pt}\fbox{\small
          #1}}} \textcolor{red}{\footnotesize #2}}
  \textcolor{blue}{#3}}
\setul{}{0.2mm}
\setstcolor{red}
\newcommand{\oldnew}[3]{{\textcolor{blue}{\setlength{\fboxsep}{1pt}\fbox{\small
        #1}}} \st{\footnotesize #2} \textcolor{blue}{#3}}

\else
\newcommand{\loldnew}[3]{#3}
\newcommand{\oldnew}[3]{#3}
\fi

\nc{\tzm}[1]{\mcomment{TZ}{#1}}
\nc{\tsm}[1]{\mcomment{TS}{#1}}
\nc{\tz}[1]{\acomment{TZ}{#1}}
\nc{\thz}[1]{\acomment{TZ}{#1}}
\nc{\ths}[1]{\acomment{TS}{#1}}

\nc{\tzon}[2][]{\oldnew{TZ}{#1}{#2}} 
\nc{\tson}[2][]{\oldnew{TS}{#1}{#2}}

\nc{\tzlon}[2][]{\loldnew{TZ}{#1}{#2}} 
\nc{\tslon}[2][]{\loldnew{TS}{#1}{#2}}

  \makeatletter{}\newcommand{\bgcolor}{black!5}

\newcommand{\substructurefillcolor}{blue!40}
\newcommand{\substructureufillcolor}{blue!60!red!40}
\newcommand{\substructuredrawcolor}{blue!80}
\newcommand{\substructureudrawcolor}{blue!60!red!80}

\newcommand{\structurefillcolor}{blue!20}
\newcommand{\structureufillcolor}{blue!60!red!20}
\newcommand{\structuredrawcolor}{blue!40}
\newcommand{\structureudrawcolor}{blue!60!red!40}

\pgfdeclarelayer{background}
\pgfdeclarelayer{substructure}
\pgfdeclarelayer{edges}
\pgfdeclarelayer{foreground}
\pgfsetlayers{background,substructure,edges,main,foreground}

\tikzstyle{background rectangle}=[
  fill=black!5,
  draw=black!20,
  inner sep=0.2cm,
  rounded corners=5pt
]

\tikzstyle{substructure}=[
  fill=\substructurefillcolor,
  draw=\substructuredrawcolor,
  inner sep=0.2cm,
  rounded corners=5pt
]
\tikzstyle{substructureu}=[
  fill=\substructureufillcolor,
  draw=\substructureudrawcolor,
  inner sep=0.2cm,
  rounded corners=5pt
]
\tikzstyle{substructurewoborder}=[
  fill=\substructurefillcolor,
  draw=\substructurefillcolor,
  inner sep=0.2cm,
  rounded corners=5pt
]

\tikzstyle{structure}=[
  fill=\structurefillcolor,
  draw=\structuredrawcolor,
  inner sep=0.2cm,
  rounded corners=5pt
]

\tikzstyle{structureu}=[
  fill=\structureufillcolor,
  draw=\structureudrawcolor,
  inner sep=0.2cm,
  rounded corners=5pt
]

\tikzstyle{mnode}=[
  circle,
  fill=black!40, 
  draw=black!80,
  minimum size=4pt, 
  inner sep=0pt
]
\tikzstyle{invisible}=[
]

\tikzstyle{invisibleEdge}=[
  transparent
]

\tikzstyle{nameNode}=[
  font=\scriptsize
]

\tikzstyle{mEdge}=[
  -latex',   semithick, 
  shorten >=3pt, 
  shorten <=3pt,
  draw=black!80,
]

\tikzstyle{delEdge}=[
  -latex', 
  semithick, 
  shorten >=3pt, 
  shorten <=3pt,
  dotted,
  draw=black!80,
]

\tikzstyle{insEdge}=[
  -latex', 
  semithick, 
  shorten >=3pt, 
  shorten <=3pt,
  dotted,
  draw=black!80,
]

\tikzstyle{dotsEdge}=[
  thick, 
  loosely dotted, 
  shorten >=7pt, 
  shorten <=7pt
]

\tikzstyle{dEdge}=[
  -latex',   thick, 
  shorten >=3pt, 
  shorten <=3pt,
  draw=black!80,
]

\tikzstyle{labelsubstruc}=[
  inner sep=0pt,
  outer sep=1.5pt,
]

\tikzstyle{labelbg}=[
  inner sep=1.5pt,
]

\newcommand{\piclogicinitsk}{
  \begin{tikzpicture}[
      xscale=0.7,
      yscale=0.7,
      font=\normalsize,
      show background rectangle
    ]
    
    \node (name) at (-1,1.5) [nameNode, font = \large] {$\state_k$:};
    
    \node (s) at (2, 1.5)[mnode, label={[labelsubstruc]above:$s$}] {};
    \node (a1) at (0,0) [mnode, label={[labelbg]below:$a_{1}$}] {};
    \node (a2) at (1.5,0) [mnode, label={[labelbg]below:$a_{k}$}] {};
    \node (t) at (2,-1.5) [mnode, label={[labelsubstruc]below:$t$}] {};

    \draw [dotsEdge] (a1) to (a2);

    \draw [mEdge] (s) to (a1);
    \draw [mEdge] (s) to (a2);

    \draw [mEdge, shorten <=9pt] (a1) to (t);
    \draw [mEdge, shorten <=9pt] (a2) to (t);

        \begin{pgfonlayer}{substructure}
      \def\f{0.7};
      \fill[substructure]  (
	    $(a1)+\f*(-1,0)$) 
	-- ($(a1)+\f*(-1,-0.5)$) 
	-- ($(t)+\f*(-0.5,-1)$) 
	-- ($(t)+\f*(1,-1)$) 
	-- ($(a2)+\f*(0.7,-0.5)$) 
	-- ($(a2)+\f*(0.7,0.5)$) 
	-- ($(s)+\f*(1,1)$) 
	-- ($(s)+\f*(-0.5,1)$) 
	-- ($(a1)+\f*(-1,0.5)$) 
	-- ($(a1)+\f*(-1,0)$);
    \end{pgfonlayer}
  \end{tikzpicture}
}

\newcommand{\piclogicinitsl}{
  \begin{tikzpicture}[
      xscale=0.7,
      yscale=0.7,
      font=\normalsize,
      show background rectangle
    ]
    
    \node (name) at (-1,1.5) [nameNode, font = \large] {$\state_l$:};
    
    \node (s) at (2, 1.5)[mnode, label={[labelsubstruc]above:$s$}] {};
    \node (a1) at (0,0) [mnode, label={[labelbg]below:$a_{1}$}] {};
    \node (a2) at (1.5,0) [mnode, label={[labelbg]below:$a_{k}$}] {};
    \node (a3) at (2.5,0) [mnode , label={[labelbg]below:$a_{k+1}$}] {};
    \node (a4) at (4,0) [mnode, label={[labelbg]below:$a_{l}$}] {};
    \node (t) at (2,-1.5) [mnode, label={[labelsubstruc]below:$t$}] {};

    \draw [dotsEdge] (a1) to (a2);
   \draw [dotsEdge] (a3) to (a4);

    \draw [mEdge] (s) to (a1);
    \draw [mEdge] (s) to (a2);
    \draw [mEdge] (s) to (a3);
    \draw [mEdge] (s) to (a4);

    \draw [mEdge, shorten <=9pt] (a1) to (t);
    \draw [mEdge, shorten <=9pt] (a2) to (t);
    \draw [mEdge, shorten <=9pt] (a3) to (t);
    \draw [mEdge, shorten <=9pt] (a4) to (t);

        \begin{pgfonlayer}{substructure}
      \def\f{0.7};
      \fill[substructure]  (
	    $(a1)+\f*(-1,0)$) 
	-- ($(a1)+\f*(-1,-0.5)$) 
	-- ($(t)+\f*(-0.5,-1)$) 
	-- ($(t)+\f*(1,-1)$) 
	-- ($(a2)+\f*(0.7,-0.5)$) 
	-- ($(a2)+\f*(0.7,0.5)$) 
	-- ($(s)+\f*(1,1)$) 
	-- ($(s)+\f*(-0.5,1)$) 
	-- ($(a1)+\f*(-1,0.5)$) 
	-- ($(a1)+\f*(-1,0)$);
    \end{pgfonlayer}
  \end{tikzpicture}
}

\newcommand{\picunarysk}{
  \begin{tikzpicture}[
      xscale=0.8,
      font=\large,
      show background rectangle,
    ]

    \node (name) at (0,1.5) [nameNode,font=\Large] {$\state'_k$:};

    \node (s) at (6, 1.5)[mnode, label={[labelsubstruc]above:$s$}] {};
    \node (1) at (0,0) [mnode, label={[labelbg]below:$a_{1}$}] {};
    \node (2) at (1.5,0) [mnode, label={[labelbg]below:$a_{i_1-1}$}] {};
    \node (3) at (2.5,0) [mnode , label={[labelsubstruc]below:$a_{i_1}$}] {};
    \node (4) at (3.5,0) [mnode, label={[labelbg]below:$a_{i_1+1}$}] {};
    \node (5) at (5,0) [mnode, label={[labelbg]below:$a_{i_2-1}$}] {};
    \node (6) at (6,0) [mnode, label={[labelsubstruc]below:$a_{i_2}$}] {};
    \node (7) at (7,0) [mnode, label={[labelbg]below:$a_{i_2+1}$}] {};
    \node (8) at (8.5,0) [mnode, label={[labelbg]below:$a_{k}$}] {};
    \node (9) at (9.5,0) [mnode, label={[labelbg]below:$a_{k+1}$}] {};
    \node (10) at (11.0,0) [mnode, label={[labelbg]below:$a_{i_k-1}$}] {};
    \node (11) at (12.0,0) [mnode, label={[labelsubstruc]below:$a_{i_k}$}] {};
    \node (12) at (13,0) [mnode, label={[labelbg]below:$a_{i_k+1}$}] {};
    \node (13) at (14.5,0) [mnode, label={[labelbg]below:$a_{l}$}] {};
    \node (14) at (15.5,0) [mnode, label={[labelbg]below:$a_{l+1}$}] {};
    \node (15) at (17.0,0) [mnode, label={[labelbg]below:$a_{m}$}] {};
    \node (t) at (6, -2.5)[mnode, label={[labelsubstruc]below:$t$}] {};

    \draw [dotsEdge] (1) to (2);
    \draw [dotsEdge] (4) to (5);
    \draw [dotsEdge] (7) to (8);
    \draw [dotsEdge] (9) to (10);
    \draw [dotsEdge] (12) to (13);
    \draw [dotsEdge] (14) to (15);

    \draw [delEdge] (s) to (14);
    \draw [delEdge] (s) to (15);

    \draw [mEdge] (s) to (1);
    \draw [mEdge] (s) to (2);
    \draw [mEdge] (s) to (3);
    \draw [mEdge] (s) to (4);
    \draw [mEdge] (s) to (5);
    \draw [mEdge] (s) to (6);
    \draw [mEdge] (s) to (7);
    \draw [mEdge] (s) to (8);
    \draw [delEdge] (s) to (9);
    \draw [delEdge] (s) to (10);
    \draw [delEdge] (s) to (11);
    \draw [delEdge] (s) to (12);
    \draw [delEdge] (s) to (13);

    \draw [mEdge, shorten <=3pt] (1) to (t);
    \draw [mEdge, shorten <=16pt] (2) to (t);
    \draw [mEdge, shorten <=18pt] (3) to (t);
    \draw [mEdge, shorten <=16pt] (4) to (t);
    \draw [mEdge, shorten <=12pt] (5) to (t);
    \draw [mEdge, shorten <=12pt] (6) to (t);
    \draw [mEdge, shorten <=12pt] (7) to (t);
    \draw [mEdge, shorten <=12pt] (8) to (t);
    \draw [mEdge, shorten <=16pt] (9) to (t);
    \draw [mEdge, shorten <=18pt] (10) to (t);
    \draw [mEdge, shorten <=29pt] (11) to (t);
    \draw [mEdge, shorten <=29pt] (12) to (t);
    \draw [mEdge, shorten <=38pt] (13) to (t);
    \draw [mEdge, shorten <=29pt] (14) to (t);
    \draw [mEdge, shorten <=39pt] (15) to (t);

        \begin{pgfonlayer}{substructure}

       \def\f{0.5};
       \fill[substructure]  
	   ($(1)+\f*(-1,0)$) 
	-- ($(1)+\f*(-1,1)$) 
	-- ($(s)+\f*(-1,1)$)
	-- ($(s)+\f*(1,1)$)
	-- ($(8)+\f*(1,1)$)
	-- ($(8)+\f*(1,-1)$)
	-- ($(t)+\f*(1,-1)$)
	-- ($(t)+\f*(-1,-1)$)
	-- ($(1)+\f*(-1,-1)$) 
	-- ($(1)+\f*(-1,0)$); 
     \end{pgfonlayer}

  \end{tikzpicture}
}

\newcommand{\picunarysl}{
  \begin{tikzpicture}[
      xscale=0.8,
      font=\large,
      show background rectangle
    ]

    \node (name) at (0,1.5) [nameNode,font=\Large] {$\state'_l$:};

    \node (s) at (6, 1.5)[mnode, label={[labelsubstruc]above:$s$}] {};
    \node (1) at (0,0) [mnode, label={[labelbg]below:$a_{1}$}] {};
    \node (2) at (1.5,0) [mnode, label={[labelbg]below:$a_{i_1-1}$}] {};
    \node (3) at (2.5,0) [mnode , label={[labelsubstruc]below:$a_{i_1}$}] {};
    \node (4) at (3.5,0) [mnode, label={[labelbg]below:$a_{i_1+1}$}] {};
    \node (5) at (5,0) [mnode, label={[labelbg]below:$a_{i_2-1}$}] {};
    \node (6) at (6,0) [mnode, label={[labelsubstruc]below:$a_{i_2}$}] {};
    \node (7) at (7,0) [mnode, label={[labelbg]below:$a_{i_2+1}$}] {};
    \node (8) at (8.5,0) [mnode, label={[labelbg]below:$a_{k}$}] {};
    \node (9) at (9.5,0) [mnode, label={[labelbg]below:$a_{k+1}$}] {};
    \node (10) at (11.0,0) [mnode, label={[labelbg]below:$a_{i_k-1}$}] {};
    \node (11) at (12.0,0) [mnode, label={[labelsubstruc]below:$a_{i_k}$}] {};
    \node (12) at (13,0) [mnode, label={[labelbg]below:$a_{i_k+1}$}] {};
    \node (13) at (14.5,0) [mnode, label={[labelbg]below:$a_{l}$}] {};
    \node (14) at (15.5,0) [mnode, label={[labelbg]below:$a_{l+1}$}] {};
    \node (15) at (17.0,0) [mnode, label={[labelbg]below:$a_{m}$}] {};
    \node (t) at (6, -2.5)[mnode, label={[labelsubstruc]below:$t$}] {};

    \draw [dotsEdge] (1) to (2);
    \draw [dotsEdge] (4) to (5);
    \draw [dotsEdge] (7) to (8);
    \draw [dotsEdge] (9) to (10);
    \draw [dotsEdge] (12) to (13);
    \draw [dotsEdge] (14) to (15);

    \draw [delEdge] (s) to (14);
    \draw [delEdge] (s) to (15);

    \draw [mEdge] (s) to (1);
    \draw [mEdge] (s) to (2);
    \draw [mEdge] (s) to (3);
    \draw [mEdge] (s) to (4);
    \draw [mEdge] (s) to (5);
    \draw [mEdge] (s) to (6);
    \draw [mEdge] (s) to (7);
    \draw [mEdge] (s) to (8);
    \draw [mEdge] (s) to (9);
    \draw [mEdge] (s) to (10);
    \draw [mEdge] (s) to (11);
    \draw [mEdge] (s) to (12);
    \draw [mEdge] (s) to (12);
    \draw [mEdge] (s) to (13);

    \draw [mEdge, shorten <=3pt] (1) to (t);
    \draw [mEdge, shorten <=16pt] (2) to (t);
    \draw [mEdge, shorten <=18pt] (3) to (t);
    \draw [mEdge, shorten <=16pt] (4) to (t);
    \draw [mEdge, shorten <=12pt] (5) to (t);
    \draw [mEdge, shorten <=12pt] (6) to (t);
    \draw [mEdge, shorten <=12pt] (7) to (t);
    \draw [mEdge, shorten <=12pt] (8) to (t);
    \draw [mEdge, shorten <=16pt] (9) to (t);
    \draw [mEdge, shorten <=18pt] (10) to (t);
    \draw [mEdge, shorten <=29pt] (11) to (t);
    \draw [mEdge, shorten <=29pt] (12) to (t);
    \draw [mEdge, shorten <=38pt] (13) to (t);
    \draw [mEdge, shorten <=29pt] (14) to (t);
    \draw [mEdge, shorten <=39pt] (15) to (t);

        \begin{pgfonlayer}{substructure}

       \def\f{0.5};
       \fill[substructure]  
	   ($(3)+\f*(-1,0)$) 
	-- ($(3)+\f*(-1,1)$) 
	-- ($(s)+\f*(-1,1)$)
	-- ($(s)+\f*(1,1)$)
	-- ($(11)+\f*(1,1)$)
	-- ($(11)+\f*(1,-1)$)
	-- ($(t)+\f*(1,-1)$)
	-- ($(t)+\f*(-1,-1)$)
	-- ($(3)+\f*(-1,-1)$) 
	-- ($(3)+\f*(-1,0)$); 

        \fill[substructure,  fill=\bgcolor]  
	   ($(6)+\f*(-1,0)$) 
	-- ($(t)+\f*(-1,0.5)$) 
	-- ($(3)+\f*(0.7,-0.5)$)
	-- ($(3)+\f*(0.7,0.5)$)
	-- ($(s)+\f*(-1,-0.5)$)
	-- ($(6)+\f*(-1,0)$); 

        \fill[substructure,  fill=\bgcolor]  
	   ($(6)+\f*(1,0)$) 
	-- ($(t)+\f*(1,0.5)$) 
	-- ($(11)+\f*(-0.7,-0.5)$)
	-- ($(11)+\f*(-0.7,0.5)$)
	-- ($(s)+\f*(1,-0.5)$)
	-- ($(6)+\f*(1,0)$);

     \end{pgfonlayer}
  \end{tikzpicture}
}

\newcommand{\picbinarytl}{
  \begin{tikzpicture}[
      xscale=0.7,
      font=\large,
      show background rectangle
    ]

    \node (name) at (0,3) [nameNode,font=\Large] {$\calT_l$:};    
    \node (s) at (6, 3)[mnode, label={[labelsubstruc]above:$s$}] {};
    \node (t) at (6, -3)[mnode, label={[labelsubstruc]below:$t$}] {};

    \node (a0) at (1.5,1.5) [mnode, label={[labelbg]above:$a_{X_1}$}]{};
    \node (a1) at (3,1.5) [mnode, label={[labelbg]above:$a_{X_k}$}]{};
    \node (a2) at (6,1.5) [mnode, label={[labelbg]above:$a_{X_l}$}]{};
    \node (a3) at (9,1.5) [invisible]{};

    \node (b1) at (0,0) [mnode, label={[labelbg]below:$b_{1}$}] {};
    \node (b2) at (1.5,0) [mnode, label={[labelbg]below:$b_{i_1-1}$}] {};
    \node (b3) at (2.5,0) [mnode , label={[labelsubstruc]below:$b_{i_1}$}] {};
    \node (b4) at (3.5,0) [mnode, label={[labelbg]below:$b_{i_1+1}$}] {};
    \node (b5) at (5,0) [mnode, label={[labelbg]below:$b_{i_2-1}$}] {};
    \node (b6) at (6,0) [mnode, label={[labelsubstruc]below:$b_{i_2}$}] {};
    \node (b7) at (7,0) [mnode, label={[labelbg]below:$b_{i_2+1}$}] {};
    \node (b8) at (8.5,0) [mnode, label={[labelbg]below:$b_{k}$}] {};
    \node (b9) at (9.5,0) [mnode, label={[labelbg]below:$b_{k+1}$}] {};
    \node (b10) at (11.0,0) [mnode, label={[labelbg]below:$b_{i_k-1}$}] {};
    \node (b11) at (12.0,0) [mnode, label={[labelsubstruc]below:$b_{i_k}$}] {};
    \node (b12) at (13,0) [mnode, label={[labelbg]below:$b_{i_k+1}$}] {};
    \node (b13) at (14.5,0) [mnode, label={[labelbg]below:$b_{l}$}] {};
    \node (b14) at (16,0) [invisible] {};

    \draw [dotsEdge] (a0) to (a1);
    \draw [dotsEdge] (a1) to (a2);
    \draw [dotsEdge,  shorten >=20pt,   shorten <=15pt,] (a2) to (a3);

    \draw [dotsEdge] (b1) to (b2);
    \draw [dotsEdge] (b4) to (b5);
    \draw [dotsEdge] (b7) to (b8);
    \draw [dotsEdge] (b9) to (b10);
    \draw [dotsEdge] (b12) to (b13);
    \draw [dotsEdge] (b13) to (b14);

    \draw [mEdge] (a0) to (b1);

    \draw [mEdge] (a1) to (b1);
    \draw [mEdge] (a1) to (b2);
    \draw [mEdge] (a1) to (b3);
    \draw [mEdge] (a1) to (b4);
    \draw [mEdge] (a1) to (b5);
    \draw [mEdge] (a1) to (b6);
    \draw [mEdge] (a1) to (b7);
    \draw [mEdge] (a1) to (b8);

    \draw [mEdge] (a2) to (b1);
    \draw [mEdge] (a2) to (b2);
    \draw [mEdge] (a2) to (b3);
    \draw [mEdge] (a2) to (b4);
    \draw [mEdge] (a2) to (b5);
    \draw [mEdge] (a2) to (b6);
    \draw [mEdge] (a2) to (b7);
    \draw [mEdge] (a2) to (b8);
    \draw [mEdge] (a2) to (b9);
    \draw [mEdge] (a2) to (b10);
    \draw [mEdge] (a2) to (b11);
    \draw [mEdge] (a2) to (b12);
    \draw [mEdge] (a2) to (b12);
    \draw [mEdge] (a2) to (b13);
        \draw [mEdge, shorten <=3pt] (b1) to (t);
    \draw [mEdge, shorten <=16pt] (b2) to (t);
    \draw [mEdge, shorten <=18pt] (b3) to (t);
    \draw [mEdge, shorten <=18pt] (b4) to (t);
    \draw [mEdge, shorten <=16pt] (b5) to (t);
    \draw [mEdge, shorten <=16pt] (b6) to (t);
    \draw [mEdge, shorten <=16pt] (b7) to (t);
    \draw [mEdge, shorten <=12pt] (b8) to (t);
    \draw [mEdge, shorten <=16pt] (b9) to (t);
    \draw [mEdge, shorten <=18pt] (b10) to (t);
    \draw [mEdge, shorten <=29pt] (b11) to (t);
    \draw [mEdge, shorten <=29pt] (b12) to (t);
    \draw [mEdge, shorten <=38pt] (b13) to (t);

        \begin{pgfonlayer}{substructure}

       \def\f{0.5};
       \fill[substructure]  
	   ($(b3)+\f*(-1,0)$) 
	-- ($(b3)+\f*(-1,1)$) 
	-- ($(a2)+\f*(-1,1)$)
	-- ($(s)+\f*(-1,1)$)
	-- ($(s)+\f*(1,1)$)
	-- ($(a2)+\f*(1,1)$)
	-- ($(b11)+\f*(1,1)$)
	-- ($(b11)+\f*(1,-1)$)
	-- ($(t)+\f*(1,-1)$)
	-- ($(t)+\f*(-1,-1)$)
	-- ($(b3)+\f*(-1,-1)$) 
	-- ($(b3)+\f*(-1,0)$); 

        \fill[substructure,  fill=\bgcolor]  
	   ($(b6)+\f*(-1,0)$) 
	-- ($(t)+\f*(-1,0.5)$) 
	-- ($(b3)+\f*(0.7,-0.5)$)
	-- ($(b3)+\f*(0.7,0.5)$)
	-- ($(a2)+\f*(-1,-0.5)$)
	-- ($(b6)+\f*(-1,0)$); 

        \fill[substructure,  fill=\bgcolor]  
	   ($(b6)+\f*(1,0)$) 
	-- ($(t)+\f*(1,0.5)$) 
	-- ($(b11)+\f*(-0.7,-0.5)$)
	-- ($(b11)+\f*(-0.7,0.5)$)
	-- ($(a2)+\f*(1,-0.5)$)
	-- ($(b6)+\f*(1,0)$);

     \end{pgfonlayer}
  \end{tikzpicture}
}

\newcommand{\picbinarytk}{
  \begin{tikzpicture}[
      xscale=0.7,
      font=\large,
      show background rectangle
    ]

    \node (name) at (0,3) [nameNode, font=\Large] {$\calT_k$:};    
    \node (s) at (6, 3)[mnode, label={[labelsubstruc]above:$s$}] {};
    \node (t) at (6, -3)[mnode, label={[labelsubstruc]below:$t$}] {};

    \node (a0) at (0,1.5) [mnode, label={[labelbg]above:$a_{X_1}$}]{};
    \node (a1) at (1.5,1.5) [mnode, label={[labelbg]above:$a_{X_k}$}]{};
    \node (a2) at (6,1.5) [mnode, label={[labelbg]above:$a_{X_l}$}]{};
    \node (a3) at (9,1.5) [invisible]{};

    \node (b1) at (0,0) [mnode, label={[labelbg]below:$b_{1}$}] {};
    \node (b2) at (1.5,0) [mnode, label={[labelbg]below:$b_{i_1-1}$}] {};
    \node (b3) at (2.5,0) [mnode , label={[labelsubstruc]below:$b_{i_1}$}] {};
    \node (b4) at (3.5,0) [mnode, label={[labelbg]below:$b_{i_1+1}$}] {};
    \node (b5) at (5,0) [mnode, label={[labelbg]below:$b_{i_2-1}$}] {};
    \node (b6) at (6,0) [mnode, label={[labelsubstruc]below:$b_{i_2}$}] {};
    \node (b7) at (7,0) [mnode, label={[labelbg]below:$b_{i_2+1}$}] {};
    \node (b8) at (8.5,0) [mnode, label={[labelbg]below:$b_{k}$}] {};
    \node (b9) at (9.5,0) [mnode, label={[labelbg]below:$b_{k+1}$}] {};
    \node (b10) at (11.0,0) [mnode, label={[labelbg]below:$b_{i_k-1}$}] {};
    \node (b11) at (12.0,0) [mnode, label={[labelsubstruc]below:$b_{i_k}$}] {};
    \node (b12) at (13,0) [mnode, label={[labelbg]below:$b_{i_k+1}$}] {};
    \node (b13) at (14.5,0) [mnode, label={[labelbg]below:$b_{l}$}] {};
    \node (b14) at (16,0) [invisible] {};

    \draw [dotsEdge] (a0) to (a1);
    \draw [dotsEdge] (a1) to (a2);
    \draw [dotsEdge,  shorten >=20pt,   shorten <=15pt,] (a2) to (a3);

    \draw [dotsEdge] (b1) to (b2);
    \draw [dotsEdge] (b4) to (b5);
    \draw [dotsEdge] (b7) to (b8);
    \draw [dotsEdge] (b9) to (b10);
    \draw [dotsEdge] (b12) to (b13);
    \draw [dotsEdge] (b13) to (b14);

    \draw [mEdge] (a0) to (b1);

    \draw [mEdge] (a1) to (b1);
    \draw [mEdge] (a1) to (b2);
    \draw [mEdge] (a1) to (b3);
    \draw [mEdge] (a1) to (b4);
    \draw [mEdge] (a1) to (b5);
    \draw [mEdge] (a1) to (b6);
    \draw [mEdge] (a1) to (b7);
    \draw [mEdge] (a1) to (b8);

    \draw [mEdge] (a2) to (b1);
    \draw [mEdge] (a2) to (b2);
    \draw [mEdge] (a2) to (b3);
    \draw [mEdge] (a2) to (b4);
    \draw [mEdge] (a2) to (b5);
    \draw [mEdge] (a2) to (b6);
    \draw [mEdge] (a2) to (b7);
    \draw [mEdge] (a2) to (b8);
    \draw [mEdge] (a2) to (b9);
    \draw [mEdge] (a2) to (b10);
    \draw [mEdge] (a2) to (b11);
    \draw [mEdge] (a2) to (b12);
    \draw [mEdge] (a2) to (b12);
    \draw [mEdge] (a2) to (b13);

        \draw [mEdge, shorten <=3pt] (b1) to (t);
    \draw [mEdge, shorten <=16pt] (b2) to (t);
    \draw [mEdge, shorten <=18pt] (b3) to (t);
    \draw [mEdge, shorten <=18pt] (b4) to (t);
    \draw [mEdge, shorten <=16pt] (b5) to (t);
    \draw [mEdge, shorten <=16pt] (b6) to (t);
    \draw [mEdge, shorten <=16pt] (b7) to (t);
    \draw [mEdge, shorten <=12pt] (b8) to (t);
    \draw [mEdge, shorten <=16pt] (b9) to (t);
    \draw [mEdge, shorten <=18pt] (b10) to (t);
    \draw [mEdge, shorten <=29pt] (b11) to (t);
    \draw [mEdge, shorten <=29pt] (b12) to (t);
    \draw [mEdge, shorten <=38pt] (b13) to (t);

        \begin{pgfonlayer}{substructure}

       \def\f{0.5};
       \fill[substructure]  
	   ($(b1)+\f*(-1,0)$) 
	-- ($(b1)+\f*(-1,1)$) 
	-- ($(a1)+\f*(-1,1)$)
	-- ($(s)+\f*(0.5,1)$)
	-- ($(s)+\f*(1.5,-0.5)$)
	-- ($(a1)+\f*(4,1)$)
	-- ($(a1)!0.5!(b5)$)
	-- ($(b8)+\f*(1,1)$)
	-- ($(b8)+\f*(1,-1)$)
	-- ($(t)+\f*(1,-1)$)
	-- ($(t)+\f*(-1,-1)$)
	-- ($(b1)+\f*(-1,-1)$) 
	-- ($(b1)+\f*(-1,0)$); 
     \end{pgfonlayer}
  \end{tikzpicture}
}

\newcommand{\picunaryqfa}{
  \begin{tikzpicture}[
      xscale=0.7,
      font=\large,
      show background rectangle
    ]

    \node (name) at (0,1.5) [nameNode, font=\large] {$\updateState{\prog}{\beta_1}{\state}$:};

    \node (s) at (5, 1)[mnode, label={[labelsubstruc]above:$s$}] {};
    \node (1) at (0,0) [mnode, label={[labelbg]below:$a_{1}$}] {};
    \node (2) at (1,0) [invisible] {};
    \node (2a) at (1,0) [invisible, font=\normalsize] {$\ldots$};

    \node (3) at (2,0) [mnode, label={[labelbg]below:$a_{i-1}$}] {};
    \node (4) at (3,0) [mnode , label={[labelsubstruc]below:$a_{i}$}] {};
    \node (5) at (4,0) [mnode, label={[labelbg]below:$a_{i+1}$}] {};

    \node (6) at (5,0) [invisible] {};
    \node (6a) at (5,0) [invisible, font=\normalsize] {$\ldots$};

    \node (7) at (6,0) [mnode, label={[labelbg]below:$a_{j-1}$}] {};
    \node (8) at (7,0) [mnode , label={[labelsubstruc]below:$a_{j}$}] {};
    \node (9) at (8,0) [mnode, label={[labelbg]below:$a_{j+1}$}] {};

    \node (10) at (9,0) [invisible] {};
    \node (10a) at (9,0) [invisible, font=\normalsize] {$\ldots$};
    \node (11) at (10,0) [mnode, label={[labelbg]below:$a_{n}$}] {};

    \node (t) at (5, -1)[mnode, label={[labelsubstruc]below:$t$}] {};

     \draw [mEdge] (1) to (2);
     \draw [mEdge] (2) to (3);
     \draw [mEdge] (3) to (4);
     \draw [mEdge] (4) to (5);
     \draw [mEdge] (5) to (6);
     \draw [mEdge] (6) to (7);
     \draw [mEdge] (7) to (8);
     \draw [mEdge] (8) to (9);
     \draw [mEdge] (9) to (10);
     \draw [mEdge] (10) to (11);

     \draw [insEdge] (s) to (4);
     \draw [insEdge] (8) to (t);

  \end{tikzpicture}
}

\newcommand{\picunaryqfb}{
  \begin{tikzpicture}[
      xscale=0.7,
      font=\large,
      show background rectangle
    ]

    \node (name) at (0,1.5) [nameNode, font=\large] {$\updateState{\prog}{\beta_2}{\state}$:};

    \node (s) at (5, 1)[mnode, label={[labelsubstruc]above:$s$}] {};
    \node (1) at (0,0) [mnode, label={[labelbg]below:$a_{1}$}] {};
    \node (2) at (1,0) [invisible] {};
    \node (2a) at (1,0) [invisible, font=\normalsize] {$\ldots$};

    \node (3) at (2,0) [mnode, label={[labelbg]below:$a_{i-1}$}] {};
    \node (4) at (3,0) [mnode , label={[labelsubstruc]below:$a_{i}$}] {};
    \node (5) at (4,0) [mnode, label={[labelbg]below:$a_{i+1}$}] {};

    \node (6) at (5,0) [invisible] {};
    \node (6a) at (5,0) [invisible, font=\normalsize] {$\ldots$};

    \node (7) at (6,0) [mnode, label={[labelbg]below:$a_{j-1}$}] {};
    \node (8) at (7,0) [mnode , label={[labelsubstruc]below:$a_{j}$}] {};
    \node (9) at (8,0) [mnode, label={[labelbg]below:$a_{j+1}$}] {};

    \node (10) at (9,0) [invisible] {};
    \node (10a) at (9,0) [invisible, font=\normalsize] {$\ldots$};
    \node (11) at (10,0) [mnode, label={[labelbg]below:$a_{n}$}] {};

    \node (t) at (5, -1)[mnode, label={[labelsubstruc]below:$t$}] {};

     \draw [mEdge] (1) to (2);
     \draw [mEdge] (2) to (3);
     \draw [mEdge] (3) to (4);
     \draw [mEdge] (4) to (5);
     \draw [mEdge] (5) to (6);
     \draw [mEdge] (6) to (7);
     \draw [mEdge] (7) to (8);
     \draw [mEdge] (8) to (9);
     \draw [mEdge] (9) to (10);
     \draw [mEdge] (10) to (11);

     \draw [insEdge] (s) to (8);
     \draw [insEdge] (4) to (t);

  \end{tikzpicture}
}

\newcommand{\piccliquea}{
  \begin{tikzpicture}[
      xscale=0.7,
      font=\large,
      show background rectangle
    ]

    \node (name) at (-0.5,3.5) [nameNode,font=\Large] {$G$:};
    
    \node (s) at (2.25, 3)[mnode, label={[labelsubstruc]above:$s$}] {};
    \node (t) at (2.25, -1.5)[mnode, label={[labelsubstruc]below:$t$}] {};

    \node (a1) at (0,1.5) [mnode, label={[labelbg]right:$a_1$}]{};
    \node (a2) at (1.5,1.5) [mnode, label={[labelbg]right:$a_2$}]{};
    \node (a3) at (3,1.5) [mnode, label={[labelbg]right:$a_3$}]{};
    \node (a4) at (4.5,1.5) [mnode, label={[labelbg]right:$a_4$}]{};

    \node (b1) at (-1,0) [mnode, label={[labelbg, label distance = 4pt]right:$b_1$}]{};
    \node (b2) at (1.33,0) [mnode, label={[labelbg]right:$b_2$}]{};
    \node (b3) at (3.66,0) [mnode, label={[labelbg]right:$b_3$}]{};
    \node (b4) at (6,0) [mnode, label={[labelbg]right:$b_4$}]{};

    \draw [mEdge] (s) to (a1);
    \draw [mEdge] (s) to (a2);
    \draw [mEdge] (s) to (a4);

    \draw [mEdge] (a2) to (b3);
    \draw [mEdge] (a3) to (b3);
    \draw [mEdge] (a3) to (b4);
    \draw [mEdge] (a4) to (b4);

    \draw [mEdge] (a2) to (b1);
    \draw [mEdge] (a2) to (b2);

    \draw [mEdge] (b1) to (t);
    \draw [mEdge] (b4) to (t);
  \end{tikzpicture}
}

\newcommand{\piccliqueb}{
  \begin{tikzpicture}[
      xscale=0.7,
      font=\large,
      show background rectangle
    ]

    \node (name) at (-0.5,3.5) [nameNode,font=\Large] {$G'$:};
    
    \node (s) at (2.25, 3)[mnode, label={[labelsubstruc, label distance = 3pt]above:$s = t$}] {};
    \node (t) at (2.25, -1.75)[invisible] {};

    \node (a1) at (0,1.5) [mnode, label={[labelbg]right:$a_1$}]{};
    \node (a2) at (1.5,1.5) [mnode, label={[labelbg]right:$a_2$}]{};
    \node (a3) at (3,1.5) [mnode, label={[labelbg]right:$a_3$}]{};
    \node (a4) at (4.5,1.5) [mnode, label={[labelbg]right:$a_4$}]{};

    \node (b1) at (-1,0) [mnode, label={[labelbg, label distance = 4pt]right:$b_1$}]{};
    \node (b2) at (1.33,0) [mnode, label={[labelbg]right:$b_2$}]{};
    \node (b3) at (3.66,0) [mnode, label={[labelbg]right:$b_3$}]{};
    \node (b4) at (6,0) [mnode, label={[labelbg]right:$b_4$}]{};

    \draw [mEdge] (s) to (a1);
    \draw [mEdge] (s) to (a2);
    \draw [mEdge] (s) to (a4);

    \draw [mEdge] (a2) to (b3);
    \draw [mEdge] (a3) to (b3);
    \draw [mEdge] (a3) to (b4);
    \draw [mEdge] (a4) to (b4);

    \draw [mEdge] (a2) to (b1);
    \draw [mEdge] (a2) to (b2);

    \draw [mEdge, bend left=50] (b1) to (s);
    \draw [mEdge, bend right=50] (b4) to (s);
  \end{tikzpicture}
}

\tikzstyle{textnode}=[
  fill=black!40, 
  draw=black!80,
  minimum size=4pt,
  rounded corners=2pt
]

\tikzstyle{relnode}=[
   font=\normalsize,
   opaque
]

\newcommand{\inittext}[1]{      \begin{minipage}{100pt}	\begin{center}
	  #1
	\end{center}
      \end{minipage}}

\newcommand{\picinitmodels}{
  \begin{tikzpicture}[
      xscale=2.2,
      yscale=0.7,
      font=\scriptsize,
      show background rectangle
    ]

    \node (eaa) at (2,2) [textnode] {
      \inittext{	empty initial database\\
	with arbitrary initialization
      }
    };
    \node (eee) at (2, 0) [textnode] {
      \inittext{	empty initial database\\
	with empty initialization \\
      }
    };

    \node (naa) at (0,2) [textnode] {
      \inittext{	non-empty initial database\\
	with arbitrary initialization
      }
    };
    \node (nei) at (0,0) [textnode] {
      \inittext{	non-empty initial database\\
	with invariant initialization      }
    };
    
    \draw[invisibleEdge] (eaa) -- node[relnode] {$=$}(naa);
    \draw[invisibleEdge] (eee) -- node[relnode] {$\subseteq$}(nei);

    \draw[invisibleEdge] (eee) -- node[relnode, rotate=90] {$\subseteq$}(eaa);
    \draw[invisibleEdge] (nei) -- node[relnode, rotate=90] {$\subseteq$}(naa);

     \node (general) at (0,-1) []{Theorem \ref{theorem:logicinit}};
     \node (binary) at (2,-1) []{Theorem \ref{theorem:binary}};

        \begin{pgfonlayer}{substructure}
       \def\f{0.55};

       \def\xa{1.6};
       \def\ya{1};
       \fill[fill=blue!20, draw=blue!30,inner sep=0.2cm,rounded corners=5pt]
	   ($(binary)+\f*(+\xa,0)$) 
	-- ($(binary)+\f*(+\xa,-\ya-0.2)$) 
	-- ($(general)+\f*(-\xa-0.1,-\ya-0.2)$) 
	-- ($(naa)+\f*(-\xa-0.1,\ya+0.1)$) 
	-- ($(eaa)+\f*(+\xa,\ya+0.1)$) 
	-- ($(binary)+\f*(+\xa,0)$); 
       \def\xa{1.61};
       \def\ya{1};
       \fill[fill=blue!40, draw=blue!50,inner sep=0.2cm,rounded corners=5pt]   
	   ($(general)+\f*(+\xa,0)$) 
	-- ($(general)+\f*(+\xa,-\ya)$) 
	-- ($(general)+\f*(-\xa,-\ya)$) 
	-- ($(nei)+\f*(-\xa,\ya)$) 
	-- ($(nei)+\f*(+\xa,\ya)$) 
	-- ($(general)+\f*(+\xa,0)$); 
    \end{pgfonlayer}

  \end{tikzpicture}
}

  \newcommand{\substructpic}[2]{
      \draw[substructure] (0,0) ellipse (0.8 and 0.8);
      \node (tmp) at (0,0.3) {#1};
      {[shift={(0,-0.3)}] #2}
   }
  \newcommand{\structpica}[1]{
       \draw[structure] (-1.0, -1.0) rectangle (1.2, 1.9);
      \node (tmp) at (0,1.3) {#1};
   }
  \newcommand{\structpicb}[1]{
      \node[structure, isosceles triangle, shape border rotate=90, minimum height=2.8cm, minimum width=3.2cm, anchor=lower side, isosceles triangle stretches
] at (0,-1.0) {};
      \node (tmp) at (0,1.3) {#1};
   }

  \newcommand{\substructpicu}[2]{
      \draw[substructureu] (0,0) ellipse (0.8 and 0.8);
      \node (tmp) at (0,0.3) {#1};
      {[shift={(0,-0.3)}] #2}
   }

  \newcommand{\structpicua}[1]{
      \draw[structureu] (-1.0, -1.0) rectangle (1.2, 1.9);
      ] at (0,-1.0) {};
      \node (tmp) at (0,1.3) {#1};
   }
  \newcommand{\structpicub}[1]{
     \node[structureu, isosceles triangle, shape border rotate=90, minimum height=2.8cm, minimum width=3.2cm, anchor=lower side, isosceles triangle stretches
] at (0,-1.0) {};
      \node (tmp) at (0,1.3) {#1};
   }

\newcommand{\picsubstructure}{
    \begin{tikzpicture}[
       xscale=0.95,
       yscale=0.9,
    ]
    \node[invisible] (tmp) at (-2,0) {};

   {[shift={(0,4.5)}] 
      \structpica{$S$}
      \node (tmp) at (-1.5,1.6){$\calS$};
      \substructpic{$A$}{\node (lo) at (0.2,0)[mnode,label={left:$\vec a$}] {};}
    }
   {[shift={(5,4.5)}] 
      \structpicb{$T$}
      \node (tmp) at (+1.2,1.6){$\calT$};
      \substructpic{$B$}{\node (ro) at (-0.4,0)[mnode,label={right:$\pi(\vec a)$}] {};}
    }
    \draw [dEdge, very thick] (lo) to node[above]{$\cong$}node[below]{$\pi$}(ro);

      \draw [dEdge](0, 3.5) -- node[right]{$\alpha = \delta(\vec a)$}(0, 2);
      \draw [dEdge](5, 3.5) -- node[right]{$\beta = \delta(\pi(\vec a))$}(5, 2); 

  {[shift={(0,0)}] 
    \structpicua{$S$}
    \node (tmp) at (-1.7,1.6){$\updateState{\prog}{\alpha}{\calS}$};
    \substructpicu{$A$}{\node (lu) at (0.2,0)[invisible]{};}
  }

  {[shift={(5,0)}] 
    \structpicub{$T$}
    \node (tmp) at (+1.2,1.6){$\updateState{\prog}{\beta}{\calT}$};
    \substructpicu{$B$}{\node (ru) at (-0.3,0)[invisible]{};}
  }
    
      \draw [dEdge, very thick] (lu) to node[above]{$\cong$}node[below]{$\pi$}(ru);

  \end{tikzpicture}
} 

\begin{document}

  \title{
    On the quantifier-free dynamic complexity of Reachability\footnote{An extended abstract of this article appeared in Proceedings of the conference Mathematical Foundations of Computer Science 2013 (MFCS 2013).} $^,$\footnote{Both authors acknowledge the financial support by the German DFG under grant SCHW 678/6-1.} }

\author{Thomas Zeume and Thomas Schwentick} 

  \maketitle

  \begin{abstract}
The dynamic complexity of the reachability query is studied in the
dynamic complexity framework of Patnaik and Immerman, restricted to quantifier-free
update formulas. 

It is shown that, with this restriction, the reachability query  cannot be
dynamically maintained,  neither with binary auxiliary relations
nor with unary auxiliary functions, and that
ternary auxiliary relations are more powerful with respect to graph
queries than binary auxiliary relations. 

Further inexpressibility results are given for the reachability query
in a different setting as well as for a syntactical restriction
of quantifier-free update formulas. Moreover inexpressibility results for some other queries are presented.
   \end{abstract}

  \section{Introduction}\label{section:intro}
    \makeatletter{}In modern data management scenarios data is subject to frequent changes. In order to avoid  costly re-computations of queries from scratch after each small modification of the data, one can try to (re-)use auxiliary data structures that have been already computed before. However, these auxiliary data structures need to be updated dynamically whenever the data changes.

The descriptive dynamic complexity framework (short: dynamic complexity) introduced by Patnaik and Immerman \cite{PatnaikI94} models this setting. It was mainly inspired by updates in relational databases. Within this framework, for a relational data\-base  subject to change, auxiliary relations are maintained with the intention to help answering a \mbox{query $\query$}. When a modification to the database, an insertion or deletion of a tuple, occurs, every auxiliary relation is updated through a first-order query (or, equivalently, through a core SQL query) that can refer to the database as well as to the auxiliary relations. A particular auxiliary relation shall always represent the answer \mbox{to  $\query$}. The class of all queries maintainable  in this way 
 is called $\DynFO$.
Beyond query or view maintenance in databases we consider it an important goal to understand the dynamic complexity of fundamental algorithmic problems. 
Reachability in directed graphs is the most intensely investigated problem in dynamic complexity (and also much studied in dynamic algorithms and other dynamic contexts) and the main query studied in this paper.  It is one of the  simplest inherently recursive queries and thus serves as a kind of drosophila in the study of the dynamic maintainability of recursive queries by non-recursive means.
It can be maintained with first-order update formulas supplemented by counting quantifiers on general graphs \cite{Hesse01} and with plain first-order update formulas on both acyclic graphs and undirected graphs \cite{PatnaikI94}. However, it is not known whether Reachability on general graphs is maintainable with first-order updates. This is one of the major open questions in dynamic complexity. 

All attempts to show that Reachability \emph{cannot} be maintained in $\DynFO$ have failed so far. In fact, there are no general inexpressibility results for $\DynFO$ at all.\footnote{Of course, a query maintainable in $\DynFO$ can be evaluated in polynomial time and thus queries that cannot be evaluated in polynomial time cannot be maintained in $\DynFO$ either.} 
This seems to be due to  a lack of understanding of the underlying mechanisms of $\DynFO$. To improve the understanding of dynamic complexity, mainly two kinds of restrictions of  $\DynFO$ have been studied: (1) limiting the information content of the auxiliary data by restricting the arity of auxiliary relations and functions and (2) reducing the amount of quantification in update formulas. 

The study of bounded arity auxiliary relations was started in \cite{DongS98} and it was shown that unary auxiliary relations are not sufficient to maintain the reachability query with first-order updates. Further inexpressibility results for unary auxiliary relations were shown and an arity hierarchy for auxiliary relations was established. However, to separate level $k$ from higher levels, database relations of arity larger than $k$ were used. Thus, a strict hierarchy has not yet been established for queries on graphs.  In \cite{DongLW03} it was shown that unary auxiliary relations are not sufficient to maintain Reachability for update formulas of any logic with certain locality properties. The proofs strongly use the ``static'' weakness of local logics and do not  fully exploit the dynamic setting, as they only require modification sequences of constant length.

The second line of research was initiated by Hesse \cite{Hesse03}. He invented and studied the class  $\DynProp$ of queries maintainable with quantifier-free update formulas. He proved that Reachability on deterministic graphs (i.e.\ graphs of unary functions) can be maintained with quantifier-free first-order update formulas. 

There is still no proof that Reachability on general graphs cannot be maintained in $\DynProp$. However, \emph{some} inexpressibility results for $\DynProp$ have been shown in \cite{GeladeMS12}: the alternating reachability query (on graphs with $\land$- and $\lor$-nodes) is not maintainable in $\DynProp$. Furthermore, on strings, $\DynProp$ exactly captures the regular languages (as Boolean queries on strings).

\paragraph{Contributions} The high-level goal of this paper is to achieve a better understanding of the dynamic maintainability of Reachability and {dynamic complexity} in general. Our main result is that the reachability query cannot be dynamically maintained by quantifier-free updates with binary auxiliary relations. This result is weaker than that of \cite{DongS98} in terms of the logic (quantifier-free vs.\ general first-order)  but it is stronger with respect to the information content of the auxiliary data (binary relations vs.\ unary relations). We  establish a strict hierarchy within  $\DynProp$ for unary, binary and ternary auxiliary relations (this is still open for $\DynFO$).

We further show that Reachability is not maintainable with unary auxiliary \emph{functions} (plus unary auxiliary relations). Although unary functions provide less information content than binary relations, they offer a very weak form of quantification in the sense that more elements of the domain can be taken into account by update formulas. 

All these results hold in the setting of Patnaik and Immerman where modification sequences start  from an empty database as well as in the setting that starts from an arbitrary database, where the auxiliary data is initialized by an arbitrary function. We show that if, in the latter setting, the initialization mapping is permutation-invariant, quantifier-free updates cannot maintain Reachability even with auxiliary functions and relations of arbitrary arity. Intuitively a permutation-invariant initialization mapping maps isomorphic databases to isomorphic auxiliary data. A particular case of permutation-invariant  initialization mappings, studied  in \cite{GraedelS12}, is when the initialization is specified by logical formulas. In this case, lower bounds for first-order update formulas have been obtained for several problems \cite{GraedelS12}.

We transfer many of our inexpressibility results to the  $\clique{k}$ query, for fixed $k \geq 3$, and the colorability query  $\colorability{k}$, for fixed $k \geq 2$. 

In \cite{ZeumeS13reachmfcs} it was shown that every query in \DynProp can be maintained by a program with negation-free quantifier-free formulas only as well as by a program with disjunction-free quantifier-free formulas only. Thus lower bounds for those syntactic fragments immediately yield lower bounds for \DynProp itself.  Here, we show that Reachability cannot be maintained by $\DynProp$ programs with update formulas that are disjunction- \emph{and} negation-free.  

A preliminary version of this work appeared in \cite{ZeumeS13reachmfcs}. It was without most of the proofs and did not contain the lower bound for disjunction- \emph{and} negation-free $\DynProp$ programs. The proofs of the normal form results obtained in \cite{ZeumeS13reachmfcs} will be included in the long version of \cite{ZeumeS14CQicdt}. The latter work establishes normal forms for variants of dynamic conjunctive queries, complementing the normal forms for \DynProp.

\paragraph{Related Work} We already described the most closely related work. As mentioned before, the reachability query has been studied in various dynamic frameworks, one of which is the Cell Probe model. In the Cell Probe model, one aims for lower bounds for the number of memory accesses of a RAM machine  for static and dynamic problems. For dynamic Reachability, lower bounds of order $\log n$ have been proved \cite{PatrascuD04}.

\paragraph{Outline} In Section \ref{section:preliminaries} we fix our notation and in Section \ref{section:setting} we define our dynamic setting more precisely. The lower bound results for Reachability are presented  in Section \ref{section:reach} (for auxiliary relations) and in Section \ref{section:dynqf} (for auxiliary functions). In Section \ref{section:moreproblems} we transfer the lower bounds to other queries. Finally, we establish a lower bound for a syntactical fragment of  $\DynProp$ in \mbox{Section \ref{section:normalforms}}. 
\paragraph{Acknowledgement} We thank Ahmet Kara and Martin Schuster for careful proofreading.

  \section{Preliminaries}\label{section:preliminaries}
    \makeatletter{}In this section, we repeat some basic notions and fix some of our notation.

A \textit{domain} is a finite set. For $k$-tuples, $\vec a = (a_1, \ldots, a_k)$ and $\vec b = (b_1, \ldots, b_k)$ over some domain $\domain$, the $2k$-tuple obtained by concatenating  $\vec a$ and $\vec b$ is denoted by $(\vec a, \vec b)$. The tuple $\vec a$ is \textit{$\norder$-ordered} with respect to an order $\norder$ of $\domain$, if $a_1 \norder \ldots \norder a_k$.  If $\pi$ is a function\footnote{Throughout this work all functions are total.} on $\domain$, we denote   $(\pi(a_1), \ldots, \pi(a_k))$ by $\pi(\vec a)$. We slightly abuse set theoretic notations and write $c \in \vec a$ if  $c = a_i$ for some $c\in\domain$ and some $i$, and $\vec a \cup \vec b$ for the set $\{a_1, \ldots, a_k, b_1\ldots, b_k\}$. 
A (relational) \emph{schema (or signature)} $\schema$ consists of a set $\relSchema$ of relation symbols and a set $\conSchema$ of constant symbols together with an arity function $\arity: \relSchema \rightarrow \N$. A \emph{database} $\db$ of schema $\schema$ with domain $\domain$ is a mapping that assigns to every relation symbol $R \in \relSchema$ a relation of arity $\arity(R)$ over $\domain$ and to every constant symbol $c \in \conSchema$ a single element (called \textit{constant}) from $\domain$. The \emph{size of a database} is the size of its domain. Unless otherwise stated (as, e.g., in \mbox{Section \ref{section:dynqf}}), we always consider relational schemas. 

A  $\schema$-\emph{structure} $\struc$ is a pair $(\domain, \db)$ where $\db$ is a database with schema $\schema$ and domain $\domain$. Sometimes we omit the schema when it is clear from the context. If $\struc$ is a structure over domain $\domain$ and $\domain'$ is a subset of $\domain$ that contains all constants of $\struc$, then the substructure of $\struc$ induced by $\domain'$ is denoted by $\restrict{\struc}{D'}$.

Let $\calS$ and $\calT$ be two structures of schema $\schema$ and over domains $S$ and $T$, respectively. A mapping \mbox{$\pi: S \mapsto T$} \textit{preserves} a relation symbol $R \in \schema$ of arity $m$, when  $\vec a \in R^\calS$ if and only if  $\pi(\vec a) \in R^\calT$, for all $m$-tuples $\vec a$. It preserves a constant symbol $c \in \schema$, if $c^\calT = \pi(c^\calS)$. The mapping is $\schema$-preserving, if it preserves all relation symbols and all constant symbols from $\schema$. Two $\schema$-structures $\calS$ and $\calT$ are \textit{isomorphic via $\pi$}, denoted by $\calS \isomorphVia{\pi} \calT$, if $\pi$ is a bijection from $S$ to $T$ which is  $\schema$-preserving. We define \mbox{$\swap{\vec a}{\vec b}: S \rightarrow S$} to be the bijection that maps, for every $i$, $a_i$ to $b_i$ and  $b_i$ to $a_i$,  and maps all other elements to themselves.

An \emph{atomic formula} is a formula of the form $R(z_1, \ldots, z_l)$ where $R$ is a relation symbol and each $z_i$ is either a variable or a constant symbol. The $k$-ary atomic type  \type{\struc}{\vec a} of a tuple $\vec a = (a_1, \ldots, a_k)$ over $\domain$ with respect to a $\schema$-structure $\struc$ is the set of all atomic formulas $\varphi(\vec x)$ with \mbox{$\vec x = (x_1, \ldots, x_k)$} for which $\varphi(\vec a)$ holds in $\struc$, where $\varphi(\vec a)$ is short for the substitution of $\vec x$ by $\vec a$ in $\varphi$. We note that the atomic formulas can use constant symbols. As we only consider atomic types in this paper, we will often simply say type instead of atomic type. The $\sigma$-type \stype{\state}{\vec a}{\sigma} is the set of atomic formulas of \type{\struc}{\vec a}  with relation symbols {from $\sigma$}. If $\norder$ is a linear order on $\domain$ we call a subset $\domain'\subseteq\domain$ \textit{$\norder$-homogeneous} (or \emph{homogeneous}, if $\norder$ is clear from the context) if, for every $l$, the type of all $\norder$-ordered $l$-tuples over $\domain'$ is the same, that is if $\type{\state}{\vec a} = \type{\state}{\vec b}$ for all ordered $l$-tuples $\vec a$ and $\vec b$.  It is easy to observe, that a set $\domain'$ is already $\norder$-homogeneous if the condition holds for every $l$ up to the maximal arity of $\schema$.

An \emph{\stgraph} is a graph $G = (V, E)$ with two distinguished nodes $s$ \mbox{and $t$}.  A \emph{$k$-layered \stgraph} $G$ is a directed graph $(V, E)$ in which
  $V-\{s,t\}$ is partitioned into $k$ layers $A_1, \ldots, A_{k}$ such that every  edge is from $s$ to $A_1$, from $A_k$ to $t$ or from $A_i$ to $A_{i+1}$, for some $i\in\{1,\ldots,k-1\}$. The \emph{reachability query} $\reachQ$ on graphs is defined as usual, that is $(a,b)$ is in $\reachQ(G)$ if $b$ can be reached from $a$ in $G$. The \emph{$s$-$t$-reachability query} $\streachQ$ is a Boolean query that is true for an \stgraph $G$, if and only if $(s,t) \in \reachQ(G)$.

Formally, an \stgraph is a structure over a schema with one binary relation symbol (interpreted by the set of edges $E$) and two constant symbols (interpreted by the two distinguished nodes $s$ and $t$).

  \section{Dynamic Queries and Programs}\label{section:setting}
    \makeatletter{}The following presentation follows \cite{WeberS07} \mbox{and \cite{GeladeMS12}}.

Informally a \emph{dynamic instance} of a static query $\query$ is a pair $(\db, \alpha)$, where $\db$ is a database and $\alpha$ is a sequence of modifications, i.e.\ a sequence of tuple insertions and deletions into $\db$. The dynamic query \dynProb{$\query$}  yields as result the relation that is obtained by first applying the modifications from $\alpha$ to $\db$ and evaluating query $\calQ$ on the resulting database. We formalize this as follows.

\begin{definition}(Abstract and concrete modifications)
  The set $\abstrUpd$ of \emph{abstract modifications} of a schema $\schema$
  contains the terms $\ins_R$ and $\del_R$, for every relation
  symbol\footnote{In this work we do not allow modification of constants,
    for simplicity.}  $R \in \schema$.  For a database $\db$ over
  schema 
                      $\schema$ with domain $\domain$, a
  \emph{concrete modification} is a term of the form $\ins_R(\vec a)$ or  $\del_R(\vec a)$ where $R \in \schema$ is a
  $k$-ary relation symbol and $\vec a$ is a $k$-tuple of elements
  from $\domain$.
\end{definition}

\textit{Applying a modification} $\ins_R(\vec a)$ to a database $\db$ replaces relation $R^\db$ by \mbox{$R^\db \cup \{\vec a\}$}. Analogously, applying a modification $\del_R(\vec a)$ replaces $R^\db$ by \mbox{$R^\db \setminus\{\vec a\}$}. 
All other relations remain unchanged. The database resulting from applying a modification $\delta$ to a database $\db$ is denoted by $\delta(\db)$. The result $\updateDB{\alpha}{\db}$ of applying a sequence of modifications $\alpha = \delta_1 \ldots \delta_m$ to a database $\db$ is defined by $\updateDB{\alpha}{\db} \df \updateDB{\delta_m}{\ldots (\updateDB{\delta_1}{\db})\ldots}$.

\begin{definition} (Dynamic Query) 
A \emph{dynamic instance} is a pair  $(\db, \alpha)$ consisting of an \emph{input database} $\db$ and a \emph{modification sequence} $\alpha$. For a  static query $\query$ with schema $\schema$, the \emph{dynamic query} \dynProb{$\query$}  is the mapping that yields $\query(\updateDB{\alpha}{\db})$, for every dynamic instance
$(\db, \alpha)$. \end{definition}

Our main interest in this work is the dynamic version $\dynProb{\streachQ}$ of the $s$-$t$-reacha\-bility query.

Dynamic programs, to be defined next, consist of an initialization mechanism and an update\footnote{In previous work (by us as well as by others) there was usually no terminological distinction between the changes that are applied to the structure at hand (e.g., database or graph) and are considered as input to an update program and the changes that are applied by an update program to the auxiliary data after such a change. Both types of changes usually have been termed \emph{updates}. In this article, we use the term \emph{modification} for changes of the database or structure and reserve the term \emph{update} for the respective change applied to the auxiliary data by the actual update program.} program. The former  yields, for every database $\db$  an 
initial state with initial auxiliary  data (and possibly with further built-in data). The latter defines the new state, for each possible modification $\delta$. The following formal definitions are illustrated in Example \ref{example:emptylist} at the end of this section.

An \emph{dynamic schema} is a triple \mbox{$(\inpSchema, \auxSchema,\builtinSchema)$} of schemas of the input database, the auxiliary database, and the built-in database and respectively. We always let $\tau\df\inpSchema\cup\auxSchema\cup\builtinSchema$. Throughout the paper, $\inpSchema$ has to be relational. In our basic setting we also require $\auxSchema$ to be relational (this will be relaxed in \mbox{Section \ref{section:dynqf}}). 

A note on the role of the built-in database is in order: as opposed to the auxiliary database, the built-in database never changes throughout a ``computation''. Our standard classes are defined over schemas without built-in databases (that is, with empty built-in schema). Built-in databases are only used to strengthen some results in one of two possible ways, (1) by showing upper bounds in which (some) auxiliary relations or functions need not be updated or (2) by showing inexpressibility results that hold for auxiliary schemas of bounded arity but with built-in relations of unbounded arity.
 In general, built-in data can be ``simulated'' by auxiliary data. However, this need not hold, e.g., if the auxiliary schema is more restricted than the built-in schema.

\begin{definition}(Update program)\label{def:updateprog}
  An \emph{update program} $P$ over dynamic schema
  \mbox{$(\inpSchema,\auxSchema,\builtinSchema)$}   is a set of first-order formulas (called \textit{update formulas} in the following) that contains,  for every $R \in \auxSchema$ and every
  abstract modification $\delta$ of some $S \in \inpSchema$, an update formula  $\uf{R}{\delta}{\vec x}{\vec y}$ over the schema $\schema$ where $\vec x$ and $\vec y$ have the same arity as $S$ and $R$, respectively.
\end{definition}

A \emph{program state} $\state$ over dynamic schema \mbox{$(\inpSchema, \auxSchema,\builtinSchema)$} is a structure $(\domain, \inp,  \aux, \builtin)$ where $D$ is the domain, $\inp$ is a database over the input schema (the \emph{current database}), $\aux$ is a database over the auxiliary schema (the \emph{auxiliary database}) and  $\builtin$ is a database over the built-in schema (the \emph{built-in database}).

The \emph{semantics of update programs} is as follows. For a modification $\delta(\vec a)$ and program state $\state=(\domain, \inp,\aux, \builtin)$ we denote by $P_\delta(\state)$ the state $(\domain, \delta(\inp), \aux', \builtin)$, where $\aux'$ consists of relations $R'\df\{\vec b \mid \state \models \uf{R}{\delta}{\vec a}{\vec b}\}$. The effect $P_\alpha(\state)$ of a modification sequence $\alpha = \delta_1 \ldots \delta_m$ to a state $\state$ is the state $\updateState{P}{\delta_m}{\ldots (\updateState{\prog}{\delta_1}{\state})\ldots}$.

\begin{definition}(Dynamic program) \label{definition:dynprog}
  A \emph{dynamic program} is a triple $(P,\init,Q)$, where
    \begin{itemize}
    \item $P$ is an update program over some dynamic schema
      \mbox{$(\inpSchema,\builtinSchema,\auxSchema)$},
    \item   the tuple $\init=(\auxInit,\builtinInit)$ consists of a function
      $\auxInit$ that maps $\inpSchema$-databases to
      $\auxSchema$-databases and a function $\builtinInit$ that maps domains
      to $\builtinSchema$-databases, and
    \item $Q\in\auxSchema$ is a designated \emph{query symbol}.
    \end{itemize}
\end{definition}

 A dynamic program $\prog=(P,\init,Q)$ \emph{maintains}  a dynamic query  \dynProb{$\query$} if, for every dynamic instance $(\db,\alpha)$, the relation $\query(\alpha(\db))$ coincides with the query relation $Q^\state$ in the state \mbox{$\state \df P_\alpha(\state_\init(\db))$}, where $\state_\init(\db)$ is the initial state, i.e.\ $\state_\init(\db) \df (\domain, \db,  \auxInit(\db), \builtinInit(D))$.

Several dynamic settings and restrictions of dynamic programs have been studied in the literature \cite{PatnaikI94, Etessami98, GraedelS12, GeladeMS12}. Possible parameters are, for instance:
\begin{itemize}
\item  the logic in which update formulas are expressed;
\item whether in dynamic instances $(\db,\alpha)$, the initial data\-base $\db$ is always empty;
\item whether the initialization mapping $\init$ is \emph{permu\-tation-invariant} (short: \textit{invariant}) in the sense that $\pi(\auxInit(\db))=\auxInit(\pi(\db))$ and \linebreak[4] $\pi(\builtinInit(\domain))=\builtinInit(\pi(\domain))$ hold, for every data\-base $\db$, domain $\domain$ and permutation $\pi$ of the domain; and
\item whether there are any built-in relations at all.
\end{itemize}

In \cite{PatnaikI97}, Dyn-FO is defined as the class of (Boolean) queries that can be maintained for empty initial data\-bases with first-order update formulas,  first-order definable initialization mapping and without built-in data. Furthermore, a larger class with polynomial-time computable initialization mapping was considered. {Also \cite{Etessami98} considers empty initial databases without built-in data.} In \cite{GraedelS12}, general instances (with non-empty initial databases) are allowed, but the initialization mapping has to be defined by logical formulas and is thus always invariant; and there is no built-in data. {In \cite{GeladeMS12} update formulas are restricted to be quantifier-free, the initial database is empty and a built-in order is available.}

In this article, the main dynamic classes do not allow built-in data. We call a dynamic schema \textit{normal} if it has an empty built-in schema $\builtinSchema$. 

We consider the following basic dynamic complexity classes.
\begin{definition}(\DynFO, \DynProp)
  $\DynFO$ is the class of all dynamic queries maintainable by dynamic
  programs with first-order update formulas over normal dynamic schemas. $\DynProp$ is the
  subclass of $\DynFO$, where update formulas do not use
  quantifiers. A dynamic program is \emph{$k$-ary} if the arity
  of its auxiliary relation symbols is at most $k$. By $k$-ary
  $\DynProp$ (resp. $\DynFO$) we refer to dynamic queries that can be
  maintained with $k$-ary dynamic programs. 
\end{definition}

At times we also consider dynamic programs with non-empty relational built-in schemas. We denote the extension of a dynamic class by programs with non-empty built-in schemas by a superscript $^*$, as in $\DynPropbi$. We note that the arity restrictions in the above definition  do not apply to the built-in relations. 

In our basic setting the initialization mappings can be arbitrary. We will explicitly state when we relax this most general setting. 
Now we sketch important relaxations. Figure \ref{figure:initmodels} illustrates the relationships between the various settings.
    \begin{figure}[t]
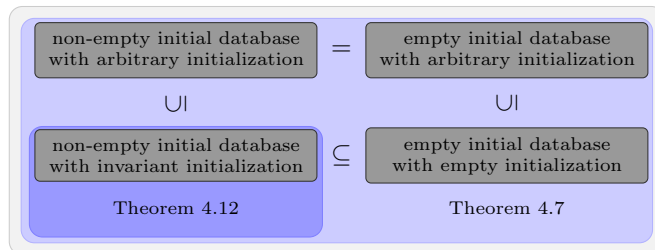
 
      \begin{center}
      \picinitmodels
      \caption{Relationship between different dynamic settings considered in the literature. Inclusion is with respect to the class of queries that can be maintained for a fixed (arbitrary) update language. Theorem \ref{theorem:binary} holds for all settings, Theorem \ref{theorem:logicinit} only for the lower left setting. \label{figure:initmodels}}
      \end{center}\vspace{-7mm}
    \end{figure}

First we note that for arbitrary initialization mappings, the same queries can be maintained regardless whether one starts from an empty or from a non-empty initial database.\footnote{The initialization for a non-empty database can be obtained as the auxiliary relations obtained after inserting all tuples of the database into the empty one.} Restricting the setting for non-empty initial databases to invariant auxiliary data initialization leads to the initialization used in \cite{GraedelS12} (called \textit{invariant initialization} in the following). For empty initial databases, allowing empty initial auxiliary data only 
leads to the initialization model of \cite{PatnaikI97, Etessami98} (called \textit{empty initialization} in the following). 

It is easy to see that applying an invariant initialization mapping to an empty database is pretty much useless, as, all tuples with the same constants at the same positions are treated in the same way. Therefore, queries maintainable in \DynFO or \DynProp with empty initial database and invariant initialization can also be maintained with empty initialization\footnote{We do not formally prove this here.}. This statement also holds in the presence of arbitrary built-in relations.

From now on we restrict our attention to quantifier-free update programs. Next, we give an example of such a program.

\begin{example}\label{example:emptylist}
    We provide a $\DynProp$-program $\prog$ for the dynamic variant of the Boolean query \problem{NonEmptySet}, where, for a unary relation $U$ subject to insertions and deletions of elements, one asks whether $U$ is empty. Of course, this query is trivially expressible in first-order logic, but not without quantifiers.    The program $\prog$ illustrates a technique to maintain lists with quantifier-free dynamic programs, introduced in \cite[Proposition 4.5]{GeladeMS12}, which is used in some of our upper bounds. 

   The program $\prog$ is over auxiliary schema  $\auxSchema = \{Q, \First, \Last, \List\}$, where $Q$ is the query bit (i.e.\ a $0$-ary relation symbol), $\First$ and $\Last$ are unary relation symbols, and $\List$ is a binary relation symbol.  The idea is to store in a program state $\state$ a list of all elements currently in $U$. The list structure is stored in the binary relation $\List^\state$ such that $\List^\state(a,b)$ holds for all elements $a$ and $b$ that are adjacent in the list. The first and last element of the list are stored in $\First^\state$ and $\Last^\state$, respectively. We note that the order in which the elements of $U$ are stored in the list depends on the order in which they are inserted into the set.  
  For a given instance of \problem{NonEmptySet} the initialization mapping initializes the auxiliary relations accordingly.

  \insertdescr{U}{a}{
    A newly inserted element is attached to the end of the list\footnote{For simplicity we assume that only elements that are not already in $U$ are inserted, the formulas given can be extended easily to the general case. Similar assumptions are made whenever necessary.}. Therefore the $\First$-relation does not change except when the first element is inserted into an empty set $U$. Furthermore, the inserted element is the new last element of the list and has a connection to the former last element. Finally, after inserting an element into $U$, the query result is 'true': 
    \begin{align*}
      \uf{\First}{\ins}{a}{x} &\df (\neg Q \mand a = x) \mor (Q \mand \First(x)) \\
      \uf{\Last}{\ins}{a}{x} &\df a = x \\
      \uf{\List}{\ins}{a}{x,y} &\df \List(x,y) \mor (\Last(x) \mand a = y)  \\
      \uf{Q}{\ins}{a}{} &\df \top. 
    \end{align*}  }  \deletedescr{U}{a}{
    How a deleted element $a$ is removed from the list, depends on whether $a$ is the first element of the list, the last element of the list or some other element of the list. The query bit remains 'true', if $a$ was not the first \emph{and} last element of the list.
    \begin{align*}
      \uf{\First}{\del}{a}{x} &\df (\First(x) \mand a \neq x) \mor (\First(a) \mand \List(a,x)) \\
      \uf{\Last}{\del}{a}{x} &\df (\Last(x) \mand a \neq x)  \mor (\Last(a) \mand \List(x,a)) \\
      \uf{\List}{\del}{a}{x,y} &\df x \neq a \mand y \neq a \mand \big(\List(x,y) \mor (\List(x, a) \mand \List(a, y))\big)\\      
      \uf{Q}{\del}{a}{} &\df \neg(\First(a) \wedge \Last(a)) 
    \end{align*}
  }
  
\end{example}

  \section{Lower Bounds for Dynamic Reachability}\label{section:reach}
    \makeatletter{}In this section we prove lower bounds for the maintainability of the dynamic $s$-$t$-reachability query \mbox{\dynstReachQ} with quantifier-free update formulas. 

First we introduce a tool for proving lower bounds for quantifier-free formulas. Afterwards we prove that 
\begin{itemize}
  \item \dynstReachQ is not in binary \DynPropbi; and
  \item \mbox{\dynstReachQ} is not in \DynPropbi with invariant initialization mappings.
\end{itemize}

The first result is used to obtain an arity hierarchy up to arity three for quantifier-free updates and binary queries.

The proofs use the following tool which is a slight variation of Lemma 1 from \cite{GeladeMS12}. The intuition is as follows. When updating an auxiliary tuple $\vec c$ after an insertion or deletion of a tuple $\vec d$, a quantifier-free update formula has access to $\vec c$, $\vec d$, and the constants only. Thus, if a sequence of modifications changes only tuples from a substructure $\calA$ of $\calS$, the auxiliary data of $\calA$  is not affected by information outside $\calA$. In particular, two isomorphic substructures $\calA$ and $\calB$ should remain isomorphic, when corresponding modifications are applied to them.

We formalize the notion of corresponding modifications as follows. Let $\pi$ be an isomorphism from a structure $\calA$ to a structure $\calB$. Two modifications $\delta(\vec a)$ on $\calA$ and $\delta(\vec b)$ on $\calB$ are said to be \textit{$\pi$-respecting} if  $\vec b = \pi(\vec a)$. Two sequences $\alpha=\delta_1\cdots\delta_m$ and $\beta=\delta'_1\cdots\delta'_m$ of modifications respect $\pi$ if, for every $i\le m$, $\delta_i$ and $\delta'_i$ are $\pi$-respecting.

  \begin{figure}[t]
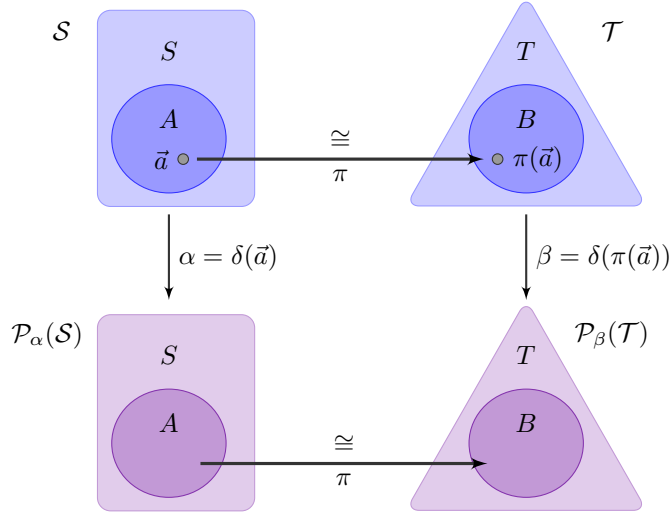
 
    \begin{center}
    \scalebox{1.0}{
     \picsubstructure
    }
    \caption{The statement of the substructure lemma. \label{figure:lemma:substructure}}
    \end{center}  \end{figure}

\begin{lemma}[Substructure lemma for \DynPropbi]\label{lemma:substruclemma}
  Let $\prog$ be a  \DynPropbi program and $\calS$ and $\calT$ states of $\prog$ with domains $S$ and $T$, respectively. Further, let $A \subseteq S$ and $B \subseteq T$ such that $\restrict{\calS}{A}$ and $\restrict{\calT}{B}$ are isomorphic via $\pi$. Then $\restrict{\updateState{P}{\alpha}{\calS}}{A}$ and $\restrict{\updateState{P}{\beta}{\calT}}{B}$ are isomorphic via  $\pi$ for all $\pi$-respecting modification sequences $\alpha$, $\beta$ on $A$ and $B$.
\end{lemma}
The substructure lemma is illustrated in Figure \ref{figure:lemma:substructure}.
\begin{proof}
  The lemma can be shown by induction on the length of the modification sequences. To this end, it is sufficient to prove the claim for a pair of $\pi$-respecting modifications $\delta(\vec a)$ and $\delta(\vec b)$ on $A$ and $B$. We abbreviate  $\restrict{\calS}{A}$ and $\restrict{\calT}{B}$ by $\calA$ and $\calB$, respectively.

   Since $\pi$ is an isomorphism from $\calA$ to $\calB$, we know that $R^{\calA}(\vec d)$ holds if and only if  $R^{\calB}(\pi(\vec d))$ holds, for every $m$-tuple $\vec d$ over $A$ and every relation symbol $R \in \schema$. Therefore, $\varphi(\vec x)$ evaluates to true in $\calA$ under $\vec d$  if and only if it does so in $\calB$ under $\pi(\vec d')$, for every quantifier-free formula $\varphi(\vec x)$ over schema $\schema$. Thus all update formulas from $\prog$ yield the same result for corresponding tuples $\vec d$ and $\pi (\vec d)$ from $A$ and $B$, respectively. Hence $\restrict{\updateState{P}{\delta(\vec a)}{\calS}}{A}$ is isomorphic to $\restrict{\updateState{P}{\delta(\pi(\vec a))}{\calS}}{B}$. This proves the claim.
\end{proof}

The following corollary is implied by Lemma \ref{lemma:substruclemma}, since the $0$-ary auxiliary relations of two isomorphic structures coincide.

\begin{corollary}\label{lemma:substruccor}
Let $\prog$ be a \DynPropbi-program with designated Boolean query symbol $Q$, and let  $\calS$ and $\calT$ be states of $\prog$ with domains $S$ and $T$. Further let $A \subseteq S$ and $B \subseteq T$ such that $\restrict{\calS}{A}$ and $\restrict{\calT}{B}$ are isomorphic via $\pi$. Then $Q$ has the same value in $\updateState{P}{\alpha}{\calS}$ and $\updateState{P}{\beta}{\calT}$ for all $\pi$-respecting sequences $\alpha$, $\beta$ of modifications on $A$ and $B$.
\end{corollary}

The Substructure Lemma can be applied along the following lines to prove that \dynstReachQ cannot be maintained in some settings with quantifier-free updates. Towards a contradiction, assume that there is a quantifier-free program \mbox{$\prog=(P, \init, Q)$} that maintains \dynstReachQ. Then, find 
\begin{itemize}
 \item two states $\calS$ and $\calT$ occurring as states\footnote{I.e.\ $\calS = \updateState{\prog}{\alpha}{\state_\init(G)}$ for some \stgraph $G$ and modification sequence $\alpha$, and likewise \mbox{for $\calT$.}} of $\prog$ with current graphs $G_\calS$ \mbox{and $G_\calT$;}
 \item substructures $\restrict{\calS}{S'}$ and $\restrict{\calT'}{T'}$ of $\calS$ and $\calT$ isomorphic via $\pi$; and
 \item two $\pi$-respecting modification sequences $\alpha$ and $\beta$ on $S'$ and $T'$ such that $\alpha(G_\calS)$ is in $\streachQ$ and $\beta(G_\calT)$  is not in $\streachQ$.
\end{itemize}
This yields the desired contradiction, since $Q$ has the same value in $\updateState{P}{\alpha}{\calS}$ and $\updateState{P}{\beta}{\calT}$ by the substructure lemma. 

How such states $\calS$ and $\calT$ can be obtained depends on the particular setting. Yet, Ramsey's theorem and Higman's lemma often prove to be useful for this task. Next, we present the variants of these theorems used in our proofs. 

\begin{theorem}[Ramsey's Theorem for Struc\-tures]\label{coro:ramsey}
 For every  schema $\schema$ and all natural numbers $k$ and $n$ there exists a number $\ramsey{\schema, k}{n}$ such that, for every $\schema$-structure $\struc$ with domain $A$ of size $\ramsey{\schema,k}{n}$, every $\vec d \in A^k$ and every order $\norder$ on $A$, there is a subset $B$ of $A$ of size $n$ with $B \cap \vec d = \emptyset$, such that,  for every $l$, the type of $(\vec a, \vec d)$ in $\struc$ is the same, for all $\norder$-ordered $l$-tuples $\vec  a$ over $B$. 
\end{theorem}

The proof of Theorem \ref{coro:ramsey} uses the well-known Ramsey theorem for hypergraphs (see, e.g., \cite[p. 7]{GrahamRS1990}) and is based on the proof of Observation 1' in \cite[p. 11]{GeladeMS12}. For the sake of completeness, the proof is presented in the following.

A \textit{$k$-hypergraph} $G$ is a pair $(V, E)$ where $V$ is a set and $E$ is a set of $k$-element subsets of $V$. If $E$ contains all $k$-element subsets of $V$, then $G$ is called \textit{complete}. A $k$-hypergraph $G'=(V',E')$ is a \textit{sub-$k$-hypergraph} of a $k$-hypergraph $G=(V,E)$, if $V' \subseteq V$ and $E'$ contains all edges $e \in E$ with $e \subseteq V'$. A \textit{$C$-coloring} $\col$ of $G$, where $C$ is a finite set of colors, is a mapping that assigns to every edge in $E$ a color from $C$, that is,  $\col: E \rightarrow C$. A \textit{$C$-colored $k$-hypergraph} is a pair $(G, \col)$ where $G$ is a $k$-hypergraph and $col$ is a $C$-coloring of $G$. If the name of the $C$-coloring is not important we also say \textit{$G$ is $C$-colored}.

\begin{theorem}(Ramsey's Theorem for Hypergraphs) \label{theorem:ramsey}
  For every set  $C$ of colors and natural numbers $n$ and $k$ there exists a number $\ramsey{C}{n}$ such that, if the edges of a complete $k$-hypergraph of size $\ramsey{C}{n}$ are $C$-colored, then the hypergraph contains a complete sub-$k$-hypergraph with $n$ nodes whose edges are all colored with the same color.
\end{theorem}

\begin{proofof}{Theorem \ref{coro:ramsey}}
    Given a schema $\schema$ and natural numbers $k$, $n$. Let $\ramsey{\schema, k}{n}$ be chosen sufficiently large     with respect to $k$, $n$, and $\schema$ such that the following argument works. Further let $\struc$ be a $\schema$-structure with domain $A$ of size greater than  $\ramsey{\schema, k}{n}$ and $\norder$ an arbitrary order on $A$. Denote by $m$ the maximal arity in $\schema$ and by $\vec c$ the constants of $\struc$ in some order. Further denote by $C$ the set of all constants and all elements occurring in $\vec d$. 

    Observe that proving the claim for $l \leq m$ is sufficient.

    We first prove the claim for $|C|=0$, by constructing inductively sets $B_l$ that satisfy the condition for $l$ with $l \leq m$. Let $B_0 = A$. The set $B_{l}$, $l \leq m$, is obtained from $B_{l-1}$ as follows. From $B_{l-1}$ a coloring $\col$ of the complete $l$-hypergraph $G$ with node set $B_{l-1}$ is constructed. The coloring $\col$ uses $l$-ary $\schema$-types as colors. An edge $e = \{e_1, \ldots, e_l\}$ with $e_1 \norder \ldots \norder e_l$ is colored by the type $\type{\struc}{e_1, \ldots, e_l}$. Because $B_{l-1}$ is large, it has, by Ramsey's theorem, a subset $B_{l}$ such that all edges $e \subseteq B_{l}$ of size $l$ are colored with the same color by $\col$. But then, by the definition of $\col$, all $\norder$-ordered $l$-tuples over $B_{l}$  have the same type in $\struc$. By this construction we obtain a set $B_m$ such that for every $l \leq m$ the type of all $\norder$-ordered $l$-tuples over $B_m$ is the same. Setting $B := B_m$ proves the claim for $|C| = 0$.

    The idea for the case $|C| \neq 0$ is to construct from $\struc$ a new structure $\struc'$ of an extended schema over domain $A' = A \setminus C$ such that $\struc'$ encodes all information about $C$ contained in $\struc$ and then use the case $|C|=0$ for $\struc'$.

    The structure $\struc'$ is of schema $\schema \cup \schema'$, where $\schema'$ contains for every $l \leq m$ and every $(l+|C|)$-ary $\schema$-type $t$, an $l$-ary relation symbol $R_t$. An $l$-tuple $\vec a$ is in $R_t^{\struc'}$ if and only if $t$ is the $\schema$-type of $(\vec a, \vec C)$.  Application of the case $|C|=0$ to $\struc'$ yields a huge homogeneous subset $B'$ with respect to $\norder$ and schema  $\schema \cup \schema'$. Then, for every $l \leq m$, the type of $(\vec a, \vec C)$ in $\struc$ is the same, for all $\norder$-ordered $l$-tuples $\vec  a$ over $B'$. This proves the claim.
\end{proofof}

Now we state the variant of Higman's Lemma that will be used later. A word $u$ is a \emph{subsequence} of a word $v$, in symbols $u \subseq v$, if $u = u_1\ldots u_k$ and $v = v_0u_1v_1\ldots v_{k-1}u_kv_k$ for some words $u_1,\ldots,u_k$ and $v_0,\ldots,v_k$. 

\begin{theorem}[Higman's Lemma]
  For every infinite sequence $(w_i)_{i \in \N}$ of words  over an alphabet $\Sigma$ there are $l$ and $k$ such that $l < k$ and $w_l \subseq w_k$.
\end{theorem}

We will actually make use of the following stronger result. See e.g.\ \cite[Proposition 2.5, page 3]{SchmitzS11} for a proof.

\begin{theorem}
  For every alphabet of size $c$ and function $g: \N \rightarrow \N$ there is a natural number $H(c)$ such that in every sequence $(w_i)_{1 \leq i \leq H(c)}$ of $H(c)$ many words with $|w_i| \leq g(i)$ there are $l$ and $k$ with $l < k$ and $w_l \subseq w_k$.
\end{theorem}

In the following we will refer to both results as Higman's Lemma.

\subsection{A Binary Lower Bound}

As already mentioned in the introduction, the proof that 
\dynstReachQ is not in unary \DynFO in \cite{DongS98}  uses constant-length modification sequences, and is mainly an application of a locality-based static lower bound for monadic second order logic. This technique does not seem to generalize to binary \DynFO. We prove the first unmaintainability result for \dynstReachQ with respect to binary auxiliary relations.
We recall that binary $\DynPropbi$ can have  built-in relations of arbitrary arity. 

\begin{theorem}\label{theorem:binary}
  $\dynProb{\streachQ}$ is not in binary $\DynPropbi$.
\end{theorem}

The proof of Theorem \ref{theorem:binary} will actually show that binary \DynPropbi cannot even maintain \dynstReachQ on 2-layered \stgraphs. 
These restricted graphs will then help us to show that binary \DynPropbi does not capture ternary \DynProp. This separation shows that the lower bound technique for binary \DynProp does not immediately transfer to ternary \DynProp (or ternary \DynPropbi). At the moment we do not know whether it is possible to adapt the technique to full \DynProp.

Before proving Theorem \ref{theorem:binary}, we show the following corresponding result for unary \DynPropbi whose proof uses the same techniques in a simpler setting.

\begin{proposition} \label{proposition:unary}
  The  dynamic \streachabilityquery is not in unary \DynPropbi, not even for $1$-layered \stgraphs.
\end{proposition}
\begin{proof}
  Towards a contradiction,  assume that  $\prog = (P, \init, \querys)$ is a dynamic program over schema $\tau = (\inpSchema, \auxSchema, \builtinSchema)$ with unary schema $\auxSchema$ that maintains the \mbox{\streachabilityquery} for 1-layered \stgraphs. Let $n'$ be sufficiently large\footnote{Explicit numbers are given at the end of the proof.} with respect to $\schema$ and $n$ be sufficiently large with respect to $n'$. Further let $m$ be the highest arity of a relation symbol from $\builtinSchema$.

    Let $G = (V, E)$  be a 1-layered \stgraph such that \mbox{$V = \{s,t\} \cup A$} with $n = |A|$ and $E = \emptyset$. Further let $\state = (V, E, \aux, \builtin)$ be the state obtained by applying $\init$ to $G$. 

Here and in the following, we do not explicitly represent the constants $s$ and $t$ in $\state$, as they never change during the application of a modification sequence (but, of course, tuples containing constants might change in the graph and in the auxiliary relations).

    First, we identify a subset of $A$ on which the built-in relations are homogeneous. By Ramsey's Theorem for structures (choosing $\vec d = (s, t)$) and because $n=|A|$ is sufficiently large with respect to $n'$ there is a set $A' \subseteq A$ of size $n'$ and an order $\norder$ on $A'$ such that all $\norder$-ordered $m$-tuples $\vec a_1$ and $\vec a_2$ over $A'$ are of equal $\builtinSchema$-type.

  Let $\state' \df (V, E', \aux', \builtin)$ be the state of $\prog$  that is reached from $\state$ after application of the following modifications to $G$ (in some arbitrary order):
    \begin{itemize}
      \item[$(\alpha)$] For every node $a \in A'$, insert edges $(s,a)$ and $(a, t)$.
    \end{itemize}
  We observe that the built-in data has not changed, but the auxiliary data might have changed.

   Let $a_1 \norder \ldots \norder a_{n'}$ be an enumeration of  the elements of  $A'$. For every $i\in\{1,\ldots,n'\}$, we define $\alpha_i$ to be the modification sequence that deletes the edges $(s, a_{n'})$, $(s, a_{{n'}-1}), \ldots, (s, a_{i+1})$, in this order. Let $\state'_{i}$ be the state reached by applying $\alpha_i$ to $\state'$. Thus, in state $\state'_{i}$ only nodes $a_1, \ldots, a_i$ have edges to node $s$. For every $i$,  we construct a word $w_i$ of length $i$, that has a letter for every node $a_1, \ldots, a_{i}$ and captures all relevant information about those nodes in $\state'_i$. The words $w_i$ are over the set of all unary types of $\auxSchema$. More precisely,    the $j$th letter $\sigma_i^j$ of $w_i$ is the unary $\auxSchema$-type of $a_j$ in $\state'_i$. We recall that the unary type of $a_j$  captures all information about the tuple $(s, a_j, t)$. 

  Since $n' = |A'|$ was chosen sufficiently large with respect to $\schema$, it follows by Higman's Lemma, that there are $k$ and $l$ such that $k < l$ and $w_k \subseq w_l$, that is, \mbox{$w_k = \sigma_k^1 \sigma_k^2 \ldots \sigma_k^{k} =  \sigma_l^{i_1}\sigma_l^{i_2} \ldots \sigma_l^{i_{k}}$} for suitable numbers $i_1 <  \ldots < i_{k}$.

  We argue that the structures $\restrict{\state'_k}\{s,t,a_1, \ldots, a_k\}$ and $\restrict{\state'_l}\{s,t,a_{i_1}, \ldots, a_{i_k}\}$ are isomorphic via the mapping $\pi$ with $\pi(a_j) = a_{i_j}$ for all $j$, $\pi(s) = s$ and $\pi(t) = t$. By definition of $A'$ and because built-in relations do not change, the mapping $\pi$ preserves $\builtinSchema$. The schema $\auxSchema$ is preserved since $a_j$ and $a_{i_j}$ are of equal unary type, by the definition of $w_k$ and $w_l$. Thus $\pi$ is indeed an isomorphism. We refer to Figure \ref{figure:proposition:unary} for an illustration.

  \begin{figure}[t]
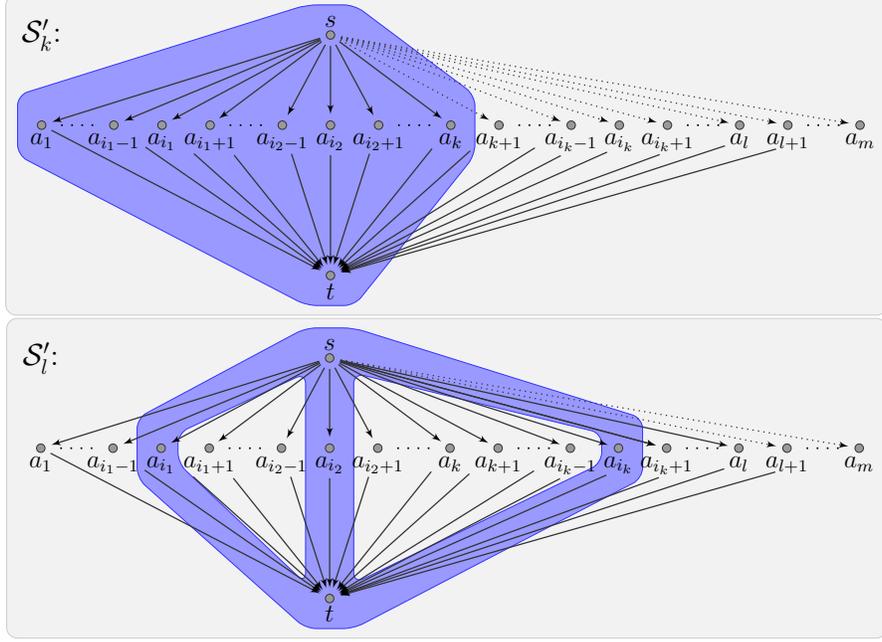
 
      \begin{center}
      \scalebox{0.8}{
    \picunarysk
      }
      \scalebox{0.8}{
	  \picunarysl
	  }
    \end{center}
      \caption{The structures $\state'_k$ and $\state'_l$ from the proof of Proposition \ref{proposition:unary}. Deleted edges are dotted. The isomorphic substructures are highlighted in blue. \label{figure:proposition:unary}}
  \end{figure}   

  Therefore, by Corollary \ref{lemma:substruccor}, the program $\prog$ computes the same query result for the following $\pi$-respecting modification sequences $\beta_1$ and $\beta_2$:
    \begin{itemize}
      \item[($\beta_1$)] Delete edges $(s, a_1), \ldots, (s, a_k)$ from $\state'_k$. 
      \item[($\beta_2$)] Delete edges $(s, a_{i_1}), \ldots, (s, a_{i_k})$ from $\state'_l$.
    \end{itemize}
  However, applying the modification sequence $\beta_1$ yields a graph where $t$ is not reachable from $s$, whereas by $\beta_2$ a graph is obtained where $t$ is reachable \mbox{from $s$} since $k < l$, the desired contradiction.

  We now specify the numbers $n$ and $n'$ that were chosen  in the beginning of the proof. In order to apply Higman's Lemma, the set $A'$ needs to be of size at least $n'\df H(|n''|)$ where $n''$ is the number of unary types of $\schema$. Therefore, the set $A$ has to be of size $n \df \ramsey{\schema}{n'}$. 
\end{proof}

Now we prove Theorem \ref{theorem:binary}, i.e.\ that \dynstReachQ is not in binary \DynPropbi. In the proof, we will again first choose a homogeneous subset with respect to the built-in relations. The notation introduced next and the following lemma prepare this step. 

We refine the notion of homogeneous sets. Let $\struc$ be a structure of some schema $\schema$ and $A$, $B$ disjoint subsets of the domain of $\struc$. We say that $B$ is \textit{$A$-$\norder$-homogeneous up to arity $m$}, if for every $l\le m$, all tuples $(a,\vec b)$, where $a \in A$ and $\vec b$ is an $\norder$-ordered $l$-tuple over $B$, have the same type. We may drop the order $\norder$  from the notation if it is clear from the context, and we may drop $A$ if $A = \emptyset$. We observe that if the maximal arity of $\schema$ is $m$ and $B$ is $A$-homogeneous up to arity $m$, then $B$ is $A$-homogeneous up to arity $m'$ for every $m'$. In this case we simply say $B$ is \textit{$A$-homogeneous}.

\begin{lemma}\label{lemma:homramsey}
  For every schema $\schema$ and natural number $n$, there is a natural number $\homramsey{\schema}{n}$ such that for any two disjoint subsets $A$, $B$ of the domain of a  $\schema$-structure $\struc$ with $|A|, |B| \geq \homramsey{\schema}{n}$, there are subsets $A' \subseteq A$ and $B' \subseteq B$ such that $|A'|, |B'| = n$ and $B'$ is $A'$-homogeneous in $\struc$.
\end{lemma}
\begin{proof}
   Let $\schema$ be a schema with maximal arity $m$. Choose $k'$ to be a large number\footnote{\label{fn:again}Again, explicit numbers can be found at the end of the proof.} with respect to $\schema$ and $n$; and let $k$ be a large number with respect to $k'$. In particular $k$ is large with respect to the number of constant symbols in $\schema$. Further let $A$, $B$ be disjoint subsets of the domain of a $\schema$-structure $\struc$ with $|A|, |B| > k$. Since  $k$ is large with respect to the number of constants in $\struc$, we assume, without loss of generality, that neither $A$ nor $B$ contains a constant.

  Fix a $k'$-tuple $\vec a = (a_1, \ldots, a_{k'})$ of $A$. Further let $\norder$ be an arbitrary order on $B$. Because $|B|$ is large with respect to $k'$, $n$ and $\schema$, and by Ramsey's theorem on structures (choose $\vec d = \vec a$), there is a subset $B'$ of $B$ of size $n$ such that  for every $l \leq m$ the type of $(\vec a, \vec b)$ in $\struc$ is the same, for all $\norder$-ordered $l$-tuples $\vec b$ \mbox{over $B'$}. 

  Since $k'$ is large with respect to $\schema$ and because there is only a bounded number of $(m+1)$-ary $\schema$-types, there is an increasing sequence $i_1, \ldots, i_n$ such that for all $l \leq m$ the $\schema$-types of tuples $(a_{i_j}, \vec b)$ are equal, for all $\norder$-ordered $l$-tuples $\vec b$ over $B'$ and $j \in \{1, \ldots, n\}$. We choose $A' := \{a_{i_1}, \ldots, a_{i_n}\}$. Then $B'$ is $A'$-homogeneous up to arity $m$ and therefore $A'$-homogeneous. 

  It remains to give explicit numbers. For the sequence $i_1, \ldots, i_n$ to exist in $1, \ldots, k'$, the number $k'$ has to be at least $n M + 1$ where $M$ is the number of $(m+1)$-ary $\schema$-types. Thus $k$ has to be at least $\ramsey{\schema,k'}{k'}+c$ where $c$ is the number of constants in $\schema$. Define $\homramsey{\schema}{n} \df k$.
\end{proof}

\begin{proofof}{Theorem \ref{theorem:binary}}
  Let us assume, towards a contradiction, that the dynamic program $(P, \init, \querys)$ over schema $\tau = (\inpSchema, \auxSchema, \builtinSchema)$ with binary $\auxSchema$ maintains the dynamic $s$-$t$-reachability query for 2-layered \stgraphs. We choose numbers $n$, $n_1$, $n_2$ and $n_3$ such that $n_3$ is sufficiently large
    with respect to $\schema$, $n_2$ is sufficiently large with respect to $n_3$, $n_2$ is sufficiently large with respect to $n_1$ and $n$ is sufficiently large with respect to $n_1$.  

  Let $G = (V, E)$ be a $2$-layered \stgraph with layers $A$, $B$, where $A$ and $B$ are both of size $n$ and \mbox{$E = \{(b, t) \mid b \in B\}$}. Further, let $\state = (V, E, \aux, \builtin)$ be the state obtained by applying $\init$ to $G$. 

  We will first choose homogeneous subsets. By Lemma \ref{lemma:homramsey} and because $n$ is sufficiently large, there are subsets $A_1$ and $B_1$ such that  $|A_1| = |B_1| = n_1$ and $B_1$ is $A_1$-$\norder$-homogeneous in $\state$, for some order $\norder$.  Next, let $A_2$ and $B_2$ be arbitrarily chosen subsets of $A_1$ and $B_1$, respectively, of size $|B_2| = n_2$ and $|A_2| = 2^{|B_2|}$, respectively. We note that $B_2$ is still $A_2$-homogeneous. In particular, $B_2$ is still $A_2$-homogeneous with respect to schema $\builtinSchema$.  We associate with every subset $X\subseteq B_2$ a unique vertex $a_X$ from $A_2$ in an arbitrary fashion. 

  Now,we define the modification sequence $\alpha$ as follows.
  \begin{itemize}
   \item[($\alpha$)] For every subset $X$ of $B_2$ and every $b \in X$ insert an edge $(a_X,b)$, in some arbitrarily chosen order.
  \end{itemize}
  Let $\state' \df (V, E', \aux', \builtin)$ be the state of $\prog$ after applying $\alpha$ to $\state$, i.e.\ $\state' = \updateState{P}{\alpha}{\state}$. We observe that the built-in data has not changed, but the auxiliary data might have changed. In particular, $B_2$ is not necessarily $A_2$-homogeneous with respect to schema $\auxSchema$ in state $\state'$.

  Our plan is to exhibit two sets $X,X'$ such that $X\subsetneq X'\subseteq B_2$ such that the restriction of $\state'$ to $\{s,t,a_{X'}\} \cup X'$ contains an isomorphic copy of $\state'$ restricted to $\{s,t,a_X\} \cup X$. Then the substructure lemma will easily give us a contradiction.

  By Ramsey's theorem and because $|B_2|$ is sufficiently large with respect to $n_2$, there is a subset $B_3 \subseteq B_2$ of size $n_3$ such that $B_3$ is $\norder$-homogeneous  in $\state'$. Let \mbox{$b_1 \norder \ldots \norder b_{n_3}$} be an enumeration of the elements of $B_3$  and let \mbox{$X_i \df \{b_1, \ldots, b_i\}$}, for every $i \in \{1, \ldots, n_3\}$.

  Let $\state'_i$ denote the restriction of $\state'$ to $X_i \cup \{s, t, a_{X_i}\}$. For every $i$,  we construct a word $w_i$ of length $i$, that has a letter for every node in $X_i$ and captures all relevant information about those nodes in $\calS'_i$. More precisely, \mbox{$w_i\df \sigma_i^1\cdots  \sigma_i^i$}, where for every $i$ and $j$, $\sigma_i^j$ is the binary type of $(a_{X_i}, b_j)$.

  Since $B_3$ is sufficiently large with respect to $\auxSchema$, it follows, by Higman's lemma, that there are $k$ and $l$ such that $k < l$ and $w_k \subseq w_l$, that is \linebreak[4] \mbox{$w_k = \sigma_k^1 \sigma_k^2 \ldots \sigma_k^{k} =  \sigma_l^{i_1}\sigma_l^{i_2} \ldots \sigma_l^{i_{k}}$} for suitable numbers $i_1 < \ldots < i_{k}$. Let \linebreak[4] \mbox{$\vec b \df (b_1, \ldots, b_k)$} and $\vec b' \df (b_{i_1}, \ldots, b_{i_k})$. Further, let $\calT_k \df \restrict{\state'_k}{T_k}$ where \mbox{$T_k = \{s,t,a_{X_k}\} \cup \vec b$}, and $\calT_l \df \restrict{\state'_l}{T_l}$ where $T_l \df \{s,t,a_{X_l}\} \cup \vec b'$.  We refer to Figure \ref{figure:theorem:binary} for an illustration of the substructures $\calT_k$ and $\calT_l$ of $\state'$.

  We show that $\calT_k \isomorph_\pi \calT_l$, where $\pi$ is the isomorphism that maps $s$ and $t$ to themselves, $a_{X_k}$ to $a_{X_l}$ and $b_j$ to $b_{i_j}$ for every $j \in \{1, \ldots, k\}$. We argue that $\pi$ fulfills the requirements of an isomorphism, for every relation symbol $R$ from $\inpSchema \cup \builtinSchema \cup \auxSchema$:
  \begin{itemize}
    \item   For the input relation $E$ this is obvious. In $\state'$ there are no edges from $s$ to nodes in $A_2$ and all nodes from $B_2$ have an edge to $t$. Further $X_l$ is connected to all nodes in $\vec b$ and $X_k$ is connected to all nodes \mbox{in $\vec b'$}.
    \item For $R \in \builtinSchema$, the requirement follows because $B_2$ is $A_2$-homo\-geneous for schema $\builtinSchema$.
    \item For $R \in \auxSchema$ of arity $2$ and two 2-tuples $\vec c$ and $\pi(\vec c)$ we distinguish two cases. First, if $\vec c$ and $\pi(\vec c)$ contain elements from $B_3$ only, then $\vec c \in R^{\calT_k}$ if and only if $\pi (\vec c) \in R^{\calT_l}$ because $B_3$ is homogeneous in $\state'$. Second, if $\vec c$ contains $s$, $t$ or $A_{X_l}$, then $\vec c \in R^{\calT_k}$ if and only if $\pi (\vec c) \in R^{\calT_l}$ because of the construction of $w_k$ and $w_l$.
  \end{itemize} 

  \begin{figure}[t]
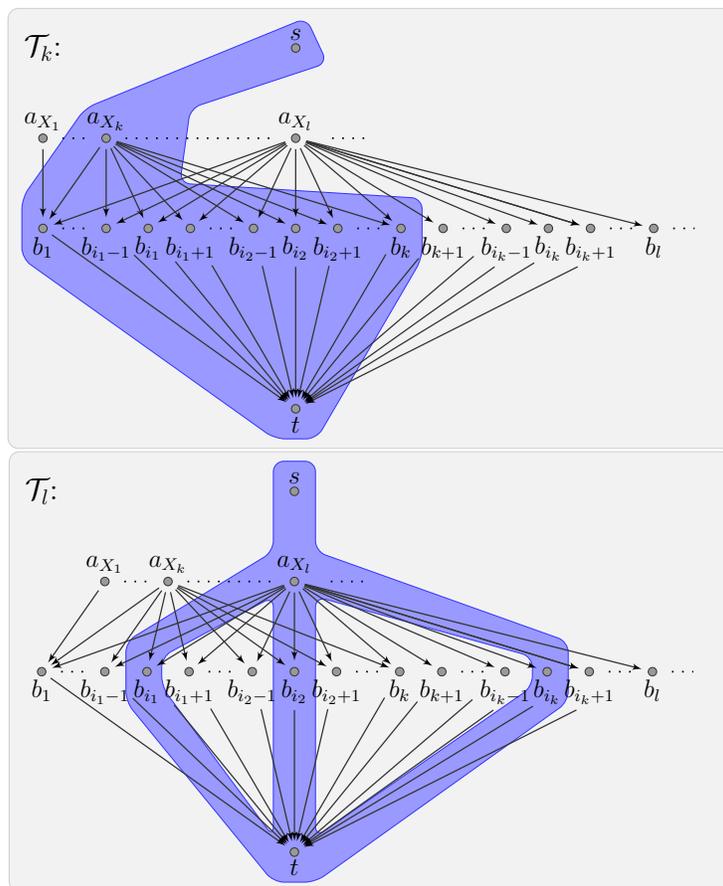
 
    \begin{center}
      \scalebox{0.8}{
      \picbinarytk
      }
      \scalebox{0.8}{
      \picbinarytl
      }
      \caption{The structure $\state'$ from the proof of Theorem \ref{theorem:binary}. The isomorphic substructures $\calT_k$ and $\calT_l$ are highlighted in blue. \label{figure:theorem:binary}}
    \end{center}  \end{figure}

  Thus, by the substructure lemma, application of the following two modification sequences to $\state'$ results in the same query result:
  \begin{itemize}
    \item[($\beta_1$)] Deleting edges $(a_{X_k}, b_1), \ldots, (a_{X_k}, b_k)$ and adding an edge $(s, a_{X_k})$.
    \item[($\beta_2$)] Deleting edges $(a_{X_l}, b_{i_1}), \ldots, (a_{X_l}, b_{i_k})$ and adding an edge $(s, a_{X_l})$.
  \end{itemize}
  However, applying $\beta_1$ yields a graph in which $t$ is not reachable from $s$, whereas by applying $\beta_2$ a graph is obtained in which $t$ is reachable from $s$. This is the desired contradiction.

  It remains to specify the sizes of the sets. To apply Higman's lemma, $|B_3|$ has to be of size at least \mbox{$n_3 \df H(m)$} where $m$ is the number of binary types \mbox{over $\auxSchema$}. Hence, for applying Ramsey's theorem, $|B_2|$ has to be of size $n_2 \df \ramsey{\schema}{n_3}$. Thus it is sufficient if $|B_1|$ and $|A_1|$ contain $n_1 \df 2^{n_2}$ elements. Therefore, by \mbox{Lemma \ref{lemma:homramsey}}, the sets $A$ and $B$ can be chosen of size $n \df \homramsey{\schema}{n_1}$.
\end{proofof}

\subsection{Separating Low Arities}

An arity hierarchy for \DynFO  was established in \cite{DongS98}. The dynamic queries $\query_{k+1}$ used to separate $k$-ary and $(k+1)$-ary \DynFO can already be maintained in $(k+1)$-ary \DynProp, thus the hierarchy transfers to \DynProp immediately. However, $\query_{k+1}$ is a $k$-ary query and has an input schema of arity $6k+1$ (improved to $3k+1$ in \cite{DongZ00}). Here we establish a strict arity hierarchy between unary, binary and ternary \DynProp for Boolean queries and binary input schemas.

We use the following problems \stTwoPath and \sTwoPath

\problemdescr{\stTwoPath}{An \stgraph $G = (V, E)$.}{Is there a path of length two from $s$ to $t$?}

\problemdescr{\sTwoPath}{A graph $G = (V, E)$ with one distinguished node $s \in V$.}{Is there a path of length two starting from $s$?}

\begin{proposition}\label{prop:unarybinary}
  The dynamic query \dynstTwoPath is in binary \DynProp, but not in unary \DynPropbi.
\end{proposition}
\begin{proofsketch}
 That \dynstTwoPath is not in unary \DynPropbi follows immediately from Proposition \ref{proposition:unary} as such a program would also maintain the dynamic $s$-$t$-reachability query for $1$-layered graphs. 

  In order to prove that \dynstTwoPath is in binary \DynProp, we sketch a \DynProp-program $(P, \init, Q)$ whose auxiliary schema contains unary relation symbols $\In$, $\Out$, $\First$, and $\Last$ and a binary relation symbol $\List$. The idea is to store, in a program state $\state$, a list of all nodes $a$ such that $(s, a, t)$ is a path in $E^\state$. The relation $\In^\state$ contains all nodes with an incoming edge from $s$, and $\Out^\state$ contains all nodes with an outgoing edge to $t$. The relations $\First^\state$, $\Last^\state$, $\List^\state$ maintain the actual list, similarly to Example \ref{example:emptylist}. The current query bit is maintained in $Q^\state$.

  For a given instance of \stTwoPath the initialization mapping initializes the auxiliary relations accordingly.

  \insertdescr{E}{(a,b)} We note that edges $(a,b)$ where $a\not=s$ and $b\not=t$ can be ignored, as they cannot contribute to any path of length 2 from $s$ to $t$. Furthermore, paths of length 2 involving only nodes $s$ and $t$ can be easily handled by $\DynProp$ formulas, and therefore will be ignored as well.

   If $a = s$ and $b\not=t$, then $b$ is inserted into $\In$, otherwise if $a\not=s$ and $b = t$ then $a$ is inserted into $\Out$. 

    Afterwards $a$ or $b$ is inserted into $\List$, if it is now     contained in both $\In$ and $\Out$.
    In that case the query  bit is set true. 

    Formally:
    \begin{align*}
      \uf{\In}{\ins}{a,b}{x} &= \In(x) \vee (x = b \mand a=s \mand b\not=s \mand b\not=t ) \\
      \uf{\Out}{\ins}{a,b}{x} &= \Out(x) \vee  (x = a \mand a\not=s \mand a\not=t \mand b=t)  \\
      \uf{\First}{\ins}{a,b}{x} &= \First(x)  \mor (\neg Q \mand \varphi_n(x))\\
      \uf{\Last}{\ins}{a,b}{x}& = (\Last(x) \mand \neg\varphi_{n}(a) \mand \neg\varphi_{n}(b)) \mor \varphi_n(x)\\
      \uf{\List}{\ins}{a,b}{x,y} &=  (\List(x,y) \mand \neg\varphi_{n}(a) \mand \neg\varphi_{n}(b))  \mor (\Last(x) \mand \varphi_n(y))\\
      \uf{Q}{\ins}{a,b}{} &= Q \vee \varphi_{n}(a) \vee \varphi_{n}(b) 
    \end{align*}
   Here, $\varphi_{n}(x)$ is an abbreviation for $$\uf{\In}{\ins}{a,b}{x} \mand \uf{\Out}{\ins}{a,b}{x} \mand(\neg \In(x) \mor\neg\Out(x))$$ expressing that $x$ is becoming newly inserted into \List.

  \deletedescr{E}{(a,b)}
    First, if $a = s$, then $b$ is removed from $\In$. Further if $b = t$ then $a$ is removed from $\Out$.

    Afterwards $a$ or $b$ is removed from $\List$, if it has been removed from $\In$ or $\Out$. If $\List$ is empty now, then the query bit is set to false. The precise formulas are along the lines of the formulas of Example \ref{example:emptylist}.
\end{proofsketch}

\begin{proposition} \label{prop:binaryternary}
  The dynamic query \dynsTwoPath is in ternary \DynProp, but not in binary \DynPropbi.
\end{proposition}

\begin{proofsketch}
  For proving that \dynsTwoPath is not in binary \DynPropbi, assume to the contrary that there is a binary \DynPropbi-program $\prog = (P, \init, Q)$ for \dynsTwoPath. With the help of $\prog$ one can, for the graphs from the proof of Proposition \ref{proposition:unary}, maintain whether there is a path from $s$ to some node of $B$. However, this yields a correct answer for \streachQ for those graphs, since in the proof all nodes of $B$ have an edge \mbox{to $t$}.

   In order to prove that \dynsTwoPath is in ternary \DynProp, we sketch a \DynProp-program $(P, \init, Q)$ whose auxiliary schema contains unary relation symbols $\In$, $\Out$, $\First_1$, $\Last_1$ and $\Empty_1$, binary relation symbols  $\List_1$, $\First_2$, $\Last_2$ and $\Empty_2$, and a ternary relation symbol $\List_2$. The idea is that in a state $\state$, the binary relation $\List^\state_1$ contains a list of all nodes $a$ on a path $(s, a, b)$ in $E^\state$, for some node $b$. The relation $\In^\state$ contains all nodes with an incoming edge from $s$,  and $\Out^\state$ contains all nodes with an outgoing edge. In order to update $\Out^\state$, the projection $\List^\state_2(a,\cdot, \cdot)$ of the ternary relation $\List^\state_2$ stores a list of nodes $b$ with \mbox{$(a, b) \in E^\state$}, for every node $a$. The lists $\List^\state_1$ and  $\List^\state_2(a, \cdot, \cdot)$ are maintained by using the technique from Example \ref{example:emptylist} and by using the auxiliary relations 
stored in $\First^\state_1$, $\Last^\state_1$, $\Empty^\state_1$, $\First^\state_2$, $\Last^\state_2$ and $\Empty^\state_2$. The current query bit is maintained in $Q^\state$.

  For a given instance of \sTwoPath the initialization mapping initializes the auxiliary relations accordingly.

  \insertdescr{E}{(a,b)}
    First, if $a = s$ then $b$ is inserted into $\In$. Otherwise, $a$ is inserted into $\Out$ and $b$ is inserted into $\List_2(a, \cdot, \cdot)$.

    Afterwards $a$ or $b$ is inserted into $\List_1$, if it is now contained in both $\In$ and $\Out$. If one of them is inserted, then the query  bit is set true.

  \deletedescr{E}{(a,b)}
    First, if $a = s$ then $b$ is removed from $\In$.  Otherwise, $b$ is removed from $\List_2(a, \cdot, \cdot)$ and if $\List_2(a, \cdot, \cdot)$ is empty afterwards, then $a$ is removed from $\Out$. 

    Afterwards $a$ or $b$ is removed from $\List_1$, if it has been removed from $\In$ or $\Out$. The query bit is set to false, if the list $\List_1$ is empty now.
\end{proofsketch}

\subsection{Invariant Initialization}

We now turn to the setting with invariant initialization. Recall that an initialization mapping $\init$ with $\init = (\auxInit, \builtinInit)$ is invariant if $$\pi(\init_{aux}(\db))=\init_{aux}(\pi(\db)) \text{ and } \pi(\builtinInit(\domain))=\builtinInit(\pi(\domain))$$ for every database $\db$, domain $\domain$ and permutation $\pi$ of the domain.
The condition $\pi(\builtinInit(\domain))=\builtinInit(\pi(\domain))$ implies that a built-in relation contains either all tuples or no tuple at all. Therefore $\DynProp$ and $\DynPropbi$ with invariant initialization mapping coincide. 

First-order logic, second-order logic and other logics considered in computer science can only define queries, i.e.\ mappings that are invariant under permutations. Therefore the following result applies, in particular, for all initialization mappings defined in those logics.

\begin{theorem}\label{theorem:logicinit}
 $\dynProb{\streachQ}$ cannot be maintained in
  \DynProp with invariant initialization mapping. This holds even for 1-layered \stgraphs.
\end{theorem}\begin{proof}
  Towards a contradiction, assume that the dynamic program $(P, \init, Q)$ with schema \mbox{$\schema = \inpSchema \cup \auxSchema$} and invariant initialization mapping $\init$ maintains the \mbox{$s$-$t$-reachability} query for 1-layered \stgraphs. Let $n$ be the number of types of tuples of arity up to $m$ for \mbox{$\auxSchema \cup \{E\}$} where $m$ is the highest arity of relation  symbols in \mbox{$\auxSchema \cup \{E\}$}.

  We consider the 1-layered \stgraphs $G_i = (V_i, E_i)$, for every $i$ from \mbox{$1,\ldots,n+1$},  with $V_i = \{s,t\} \cup A_i$ where $A_i = \{a_0, \ldots, a_{i}\}$ and $E = \{s\} \times A_i \cup A_i \times \{t\}$. Further, we let $\state_i = (V_i, E_i, \aux_i)$ be the state obtained by applying $\init$ to $G_i$.

  Our goal is to find $\state_k$ and $\state_l$ with $k<l$ such that $\state_k$ is isomorphic to $\restrict{\state_l}{V_k}$ (see Figure \ref{figure:theorem:logicinit} for an illustration). Then, by the substructure lemma, the program $\prog$ computes the same query result for the following modification sequences:
    \begin{itemize}
      \item[($\beta_1$)] Delete edges $(s, a_{0}), \ldots, (s, a_{k})$ from $\state_k$. 
      \item[($\beta_2$)] Delete edges $(s, a_{0}), \ldots, (s, a_{k})$ from $\state_l$.
    \end{itemize}
  However, applying the modification sequence $\beta_1$ yields a graph where $t$ is reachable from $s$, whereas by $\beta_2$ a graph is obtained where $t$ is not reachable from $s$,  a contradiction.

  Thus it remains to find such states $\state_k$ and $\state_l$. A tuple is \textit{diverse}, if all components are pairwise different. For arbitrary $m' \leq m$, diverse tuples \mbox{$\vec a, \vec b \in A^{m'}$} and $i \leq n$, we observe that $G_i \isomorphVia{\swap{\vec a}{\vec b}} G_i$ where $\swap{\vec a}{\vec b}$ is the bijection that maps $a_i$ to $b_i$, $b_i$ to $a_i$ and every other element from $S$ to itself. Therefore $\state_i \isomorphVia{\swap{\vec a}{\vec b}} \state_i$ by the invariance of $\init$. Thus $\type{\state_i}{\vec a} = \type{\state_i}{\vec b}$, and therefore all diverse $m'$ tuples are of the same type in $\state_i$.

  Since $n$ is the number of types up to arity $m$, there are two states $\state_k$ and $\state_l$ such that, for every $m' \leq m$, all diverse $m'$-tuples are of the same type in $\state_k$ and $\state_l$. But then  $\state_k \isomorph \restrict{\state_l}{V_k}$.
    \begin{figure}[t]
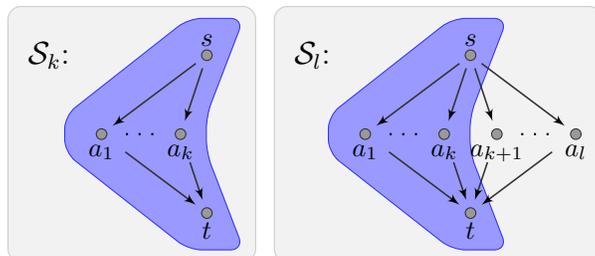
 
  \begin{center}
      \piclogicinitsk
      \piclogicinitsl

      \caption{The structures $\state_k$ and $\state_l$ from the proof of Theorem \ref{theorem:logicinit}. The isomorphic substructures are highlighted in blue.  \label{figure:theorem:logicinit}}
  \end{center}
    \end{figure}
\end{proof}

The proof of the previous result does not extend to \DynFO, since reachability in graphs of depth three is expressible even in (static) predicate logic. The proof fails, because the substructure lemma does not hold for $\DynFO$-programs.  At first glance, layered graphs with many layers look like a good candidate for proving that \DynFO cannot maintain \streachp in this setting. However, in \cite{GraedelS12} it is shown that $\DynFO$ with FO+TC-definable initialization mappings can express \streachp for arbitrary acyclic graphs.

  \section{Lower Bounds with Auxiliary Functions}\label{section:dynqf}
      \makeatletter{}In this section we consider the extension of the quantifier-free update formalism by auxiliary functions. Recall that $\DynProp$-update formulas can only access  the inserted or deleted tuple $\vec a$ and the currently updated tuple $\vec b$ of an auxiliary relation. With auxiliary functions further elements might be accessed via function terms over $\vec a$ and $\vec b$. Thus, in a sense, auxiliary functions can be seen as adding weak quantification to quantifier-free formulas. The class  of dynamic queries that can be maintained with quantifier-free update formulas and auxiliary functions is denoted \DynQF.

After the formal definition of \DynQF and adapting the substructure lemma to it, we prove that
\begin{itemize}
  \item \dynstReachQ is not in unary \DynQF; and
  \item \dynstReachQ is not in \DynQF with invariant initialization.
\end{itemize} 

When full first-order updates are available, auxiliary functions can be simulated in a straight forward way by auxiliary relations. However, without quantifiers this is not possible. Auxiliary functions are quite powerful. While only regular languages can be maintained in \DynProp, all Dyck languages, among other non-regular languages, can be maintained in \DynQF \cite{GeladeMS12}. Furthermore, undirected reachability can be maintained in \DynQF with built-in relations \cite{Hesse03}.

We extend our definition of schemata to allow also function symbols. Within this section, a \emph{schema (or signature)} $\schema$ consists of a set $\relSchema$ of relation symbols, a set $\funSchema$ of function symbols and a set $\conSchema$ of constant symbols together with an arity function $\arity: \relSchema\cup\funSchema \mapsto \N$. A schema is \emph{relational} if $\funSchema=\emptyset$.  A \emph{database} $\db$ of schema $\schema$ with domain $\domain$ is a mapping that assigns to every relation symbol $R \in \relSchema$ a relation of arity $\arity(R)$ over $\domain$,  to every $k$-ary function symbol $f \in \funSchema$ a $k$-ary function, and to every constant symbol $c \in \conSchema$ a single element (called \textit{constant}) from $\domain$. 

In the following, we extend our definition of update programs for the
case of auxiliary schemas with functions\footnote{We also allow functions in built-in schemas. As they are not updated they do not need any further particular definitions.}. It is straightforward to
extend the definition of update formulas for auxiliary relations: they
simply can make use of function terms. 
However,  following the spirit
of \DynProp, we allow a more powerful update mechanism for auxiliary
functions that allows case distinctions in addition to composition of
function terms.

The following definitions are adapted from \cite{GeladeMS12}. 
\begin{definition}(Update term)
  \textit{Update terms} are inductively defined by the following.
  \begin{itemize}
  \item[(1)] Every variable and every constant is an update term.
  \item[(2)] If $f$ is a $k$-ary function symbol and $t_1, \dots, t_k$
    are update terms, then $f(t_1, \ldots, t_k)$ is an update term.
  \item[(3)] If $\phi$ is a quantifier-free update formula
    (possibly using update terms) and $t_1$ and $t_2$ are update
    terms, then $\ite{\phi}{t_1}{t_2}$ is an update term.
  \end{itemize} 
The semantics of update terms associates with every
  update term $t$ and interpretation $I=(\state,\beta)$, where $\state$
  is a state and $\beta$ a variable assignment, a value $\sem{t}{I}$ from
  $S$.  
  The semantics of (1) and (2) is straightforward. If $\state\models\phi$ holds, then $\sem{\ite{\phi}{t_1}{t_2}}{I}$
  is  $\sem{t_1}{I}$, otherwise $\sem{t_2}{I}$.
\end{definition}

The extension of the notion of update programs for auxiliary schemas with function symbols is now straightforward. An update program still has an update \mbox{formula $\ufwa{R}{\delta}$} (possibly using terms built from function symbols) for every relation symbol $R \in \auxSchema$ and every
  abstract modification $\delta$. Furthermore, it has, for every
  abstract modification $\delta$ and  every function symbol $f \in \auxSchema$, an update term $\ut{f}{\delta}{\vec x}{\vec y}$. For a concrete modification $\delta(\vec a)$ it redefines $f$ for each tuple $\vec b$ by evaluating  $t^f_\delta(\vec a;\vec b)$ in the current state.  
 
  \begin{definition}(\DynQF)
    \DynQF is the class of queries maintainable by quantifier-free
    update programs with (possibly) auxiliary functions. The class $k$-ary \DynQF is defined via update programs that use  auxiliary functions and relations 
    of arity at most $k$.
  \end{definition}
We define $\DynQFbi$ as the extension of $\DynQF$ with built-in functions and relations of arbitrary arity.

Lists can be  represented by unary functions in a straightforward
way. Therefore, it is not surprising that the upper bound of
Proposition \ref{prop:unarybinary} already holds for unary \DynProp with unary built-in functions.

\begin{proposition} \label{proposition:onelayereddynqf}
  $\dynProb{\streachQ}$ on 1-layered \stgraphs can be maintained in unary  \DynQFbi with relational auxiliary schema and only unary built-in functions. In particular, $\dynProb{\streachQ}$ on 1-layered \stgraphs can be maintained in unary  \DynQF.
\end{proposition}

\newcommand{\So}{\text{ConS}}
\newcommand{\Ta}{\text{ConT}}

\begin{proofsketch}
    We construct a $\DynQFbi$-program $\prog$ over relational auxiliary schema $\{Q, \So, \Ta, C\}$ and functional built-in schema $\{\Pred, \Succ\}$, where $Q$ is the query bit (i.e.\ a $0$-ary relation symbol), $\So$, $\Ta$ and $C$ are unary relation symbols and $\Pred$ and $\Succ$ are unary function symbols. 
    
  The basic idea is to interpret elements of $\domain$ as numbers according to their position in the graph of $\Succ$. For simplicity, but without loss of generality, we therefore assume that the domain is of the form $\domain = \{0, \ldots, n-1\}$ with $s = 0$ and $t= n-1$. For every state \state, the built-in function $\Succ^\state$ is then the standard successor function on $\domain$ (with $\Succ^\state(n-1) = n-1$) and $\Pred^\state$ is its corresponding predecessor function (with $\Pred^\state(0) = 0$).

 The second idea is to
store the current number $i$ of vertices connected to both $s$ and $t$ by letting $C^\state=\{i\}$. If an edge-insertion connects an element to $s$ and $t$ then $i$ is replaced by $i+1$ in $C^\state$ with the help of $\Pred^\state$ and $\Succ^\state$. Analogously $i$ is replaced by $i-1$ for edge-removals that disconnect an element from $s$ or $t$. The relations $\So^\state$ and $\Ta^\state$ store the elements currently connected to $s$ and $t$, respectively.

    For a given instance of the $s$-$t$-reachability query on 1-layered \stgraphs the initialization mapping initializes the auxiliary relations accordingly.

  \insertdescr{E}{(a,b)}{
    If $a=s$ then node $b$ is inserted into $\So$; if $b=t$ then node $a$ is inserted into $\Ta$. Further, if $a$ or $b$ is now in both $S$ and $T$ then the counter is incremented by $1$:
    \begin{align*}
      \uf{\So}{\ins}{a,b}{x} &\df (a = s \wedge x = b) \vee \So(x) \\
      \uf{\Ta}{\ins}{a,b}{x} &\df (b = t \wedge x = a) \vee \Ta(x) \\
      \uf{C}{\ins}{a,b}{x} &\df \big(a = s \wedge \Ta(b) \wedge C(\Pred(x))\big) \\
      &  \quad \quad \vee \big(b = t \wedge \So(a) \wedge C(\Pred(x))\big)  \\
      & \quad \quad \vee \big(a = s \wedge \neg \Ta(b) \wedge C(x)\big) \\
      & \quad \quad \vee \big(b = t \wedge \neg \So(a) \wedge C(x)\big)  \\
      \uf{Q}{\ins}{a,b}{} &\df \neg \uf{C}{\ins}{a,b}{s}
    \end{align*}
  }

  Deletions can be maintained in a similar way.
\end{proofsketch}

We refer to \cite[Section 4.3]{Hesse03} and \cite[Sections 4 and 6]{GeladeMS12} for more examples of \DynQF-programs.

In the following we work towards lower bounds for \DynQF. We first extend the substructure lemma to non-relational structures. If a modification changes a tuple from a substructure $\calA$ of a structure $\calS$, then the update of the auxiliary data of $\calA$ can depend on elements obtained from applying functions to elements in $\calA$. We formally capture these elements by the notion of neighborhood, defined next.

The\textit{ nesting depth} $\nd(t)$ of an update term $t$ is its
  nesting depth with respect to function symbols: If $t$ is a variable, then $\nd(t) = 0$; if $t$ is of
  the form $f(t_1, \ldots, t_k)$ then $\nd(t) = \max\{\nd(t_1),
  \ldots,\nd(t_k)\}+1$; and if $t$ is of the form
  $\ite{\phi}{t_1}{t_2}$ then $\nd(t) = \max\{\nd(\phi), \nd(t_1),
  \nd(t_2)\}$. The nesting depth $\nd(\phi)$ of $\phi$ is the
  maximal nesting depth of all update terms occurring in $\phi$. The \textit{nesting depth of $\prog$} is the maximal
  nesting depth of an update term occurring in $\prog$.

For a schema $\schema$, let $\Terms{\schema}{k}$ be the set of terms of nesting depth at most $k$ with function symbols from $\schema$. Informally, the $k$-neighborhood of a set $A$ is the set of all elements of $S$ that can be obtained by applying a term of nesting depth at most $k$ to a vector of elements from $A$.

\begin{definition}(Neighborhoods)
Let $\state$ be a state with domain $S$ over schema $\schema$ and $k\ge 0$. The \textit{$k$-neighborhood}
  $\nb{A}{\state}{k}$ of a set $A\subseteq S$ is the set 
\[
\{\sem{t}{(\state,\beta)}\mid t\in \Terms{\schema}{k}\text{ and } \beta(x)\in A, \text{for every variable $x$ in $t$}\}.
\]
A subset $A$
  of $S$ is \textit{closed} if $\nb{A}{\state}{1} = A$.
\end{definition}
The $k$-neighborhood of a tuple $\vec a$ or a single element $a$ is defined accordingly. We note that for a closed set $A$ it also holds $\nb{A}{\state}{k} = A$, for every $k$.

A bijection $\pi$ between (the domains  $S$ and $T$ of) two structures $\calS$ and $\calT$ over $\schema = \relSchema \cup \funSchema$ is an \emph{isomorphism}, if it preserves $\relSchema$ and \mbox{$\pi(f^\state(\vec a))= f^\calT(\pi(\vec a))$} for all $k$-ary function symbols $f \in \funSchema$ and $k$-tuples $\vec a$ over $S$. Two subsets $A \subseteq S$, $B \subseteq T$ are \textit{$k$-similar}, if there is a  bijection \mbox{$\pi: \nb{A}{\calS}{k} \rightarrow \nb{B}{\calT}{k}$} such that
\begin{itemize}
 \item the restriction of $\pi$ to $A$ is a bijection of $A$ and $B$,
  \item $\pi$ satisfies $\pi(t^\state(\vec a)) := t^\calT(\pi(\vec a))$ for all $t \in \Terms{\funSchema}{k}$ and $\vec a$ over $A$, and
  \item $\pi$ preserves $\relSchema$ on $\nb{A}{\calS}{k}$.
\end{itemize}
We write $A\approx_k^{\pi,\calS,\calT} B$ to indicate that $A$ and $B$ are $k$-similar via $\pi$ in $\calS$ and $\calT$.
We drop $\calS$ and $\calT$ from this notation if they are clear from the context, and we drop $\pi$ if the name is not important. We also write 
$(a_1,\ldots,a_p)\approx_k^{\calS,\calT} (b_1,\ldots,b_p)$ to indicate that $\{a_1,\ldots,a_p\}\approx_k^{\pi,\calS,\calT} \{b_1,\ldots,b_p\}$ via the isomorphism $\pi$ that maps $a_i$ to $b_i$, for every $i\in\{1,\ldots,p\}$.
Note that if $A\approx_0 B$, then $\restrict{\calS}{A}$ and $\restrict{\calT}{B}$ are $\relSchema$-isomorphic by the first and third property.

The following lemma is a slight generalization of Lemma 4 from \cite{GeladeMS12} and a generalization of the substructure lemma for \DynProp (Lemma \ref{lemma:substruclemma}) to $\DynQFbi$. Intuitively, the substructure lemma for \DynQFbi requires not only similarity of the substructures but of their neighborhoods as well.

\begin{lemma}[Substructure lemma for \DynQF]\label{lemma:substruclemmafun}
  Let $\prog$ be a \DynQFbi program with nesting depth $k$ and let $l$ be some number. Furthermore let $\calS$ and $\calT$ be states of $\prog$  with domains $S$ and $T$ and let $A$ and $B$ be subsets of $S$ and $T$, respectively. There is a number $m \in \N$ such that if $A\approx_{m}^{\pi,\calS,\calT} B$, then $A\approx_0^{\pi,\updateState{P}{\alpha}{\calS},\updateState{P}{\beta}{\calT}} B$, for all $\pi$-respecting modification sequences $\alpha$ and $\beta$  on $A$ and $B$ of length at most $l$.
\end{lemma}
\begin{proof}
  The proof is an extension of the proof of Lemma \ref{lemma:substruclemma}.  The lemma follows by an induction over the length $l$ of the modification sequence. For $l=0$ there is nothing to prove. The induction step follows easily using Claim (C) below.

  Let $\delta(\vec a)$ and $\delta(\vec b)$ be two $\pi$-respecting modifications on $A$ and $B$, respectively, i.e.\ $\vec b = \pi (\vec a)$. Let   $\calS'\df\updateState{P}{\delta(\vec a)}{\calS}$ and $\calT'\df\updateState{P}{\delta(\vec b)}{\calT}$. We prove the following claims for arbitrary $r \in \N$:
  \begin{enumerate}
   \item[(A)]     If $A\approx_{r+k}^{\pi,\calS,\calT} B$, then\footnote{Of course, the following two statements also hold for relation and function symbols \mbox{from $\builtinSchema$.}} for all $\vec c$ over $\nb{A}{\calS}{r}$:
      \begin{itemize}
      \item[(i)] $\vec c \in R^{\calS'}$ if and only if $\pi(\vec c) \in R^{\calT'}$ for all relation symbols $R \in \auxSchema$.
      \item[(ii)]  $f^{\calS'}(\vec c) \in \nb{A}{\calS}{r+k}$ and $\pi(f^{\calS'}(\vec c)) = f^{\calT'}(\pi(\vec c))$ for all function symbols $f \in \auxSchema$.
      \end{itemize}
    \item[(B)]    If $A\approx_{r \cdot k}^{\pi,\calS,\calT} B$, then $t^{\calS'}(\vec c) \in \nb{A}{\calS}{r\cdot k}$ and  $\pi(t^{\calS'}(\vec c)) = t^{\calT'}(\pi(\vec c))$ for all terms $t \in \Terms{\auxSchema\cup\builtinSchema}{r}$ and $\vec c$ over $S$.
    \item[(C)]  If $A\approx_{r \cdot k+k}^{\pi,\calS,\calT} B$, then $A\approx_{r}^{\pi,\calS',\calT'} B$.
  \end{enumerate}

 We prove Claim (A) first. We recall that $\vec c \in R^{\calS'}$ if and only if $\state\models\uf{R}{\delta}{\vec a}{\vec c}$, and that $f^{\calS'}(\vec c)$ is $\sem{\ut{f}{\delta}{\vec x}{\vec y}}{(\state,\gamma)}$, where $\gamma$ maps $(\vec x,\vec y)$ to $(\vec a, \vec c)$. Since $\vec a$ and $\vec c$ are tuples over $\nb{A}{\state}{r}$ it is sufficient to prove, for every tuple $\vec d$ over $\nb{A}{\state}{r}$, that (i) $\varphi(\vec d)$ holds in $\state$ if and only if $\varphi(\pi(\vec d))$ holds in $\calT$, for every quantifier-free formula $\varphi$ with nesting depth at most $k$,  and that\footnote{Here, we use $\vec d$ to denote the variable assignment mapping the free variables of $t$ to the components of $\vec d$.} (ii) \mbox{$\pi(\sem{t}{(\state,\vec d)}) = \sem{t}{(\calT,\pi(\vec d))}$}, for every update term $t$ with nesting depth at most $k$. 

  The proof is by induction on $k$. We start with the base case. If $k=0$, terms and update terms do not use any function symbols and therefore, (i) and (ii) hold trivially, because $\pi$ witnesses the $(r+k)$-similarity of $A$ and $B$ in $\state$ and $\calT$. 
  For the induction step, we consider update terms and update formulas with nesting depth $k'\in\{1,\ldots,k\}$. If an update term $t$ with $\nd(t) = k'$ is of the form $f(\vec s)$ with $\vec s = (s_1, \ldots, s_n)$, then, by induction hypothesis, 
\mbox{$\pi(\sem{s_i}{(\state,\vec e_i)}) = \sem{s_i}{(\calT,\pi(\vec e_i))}$} and $s_i^\state(\vec e_i) \in \nb{A}{\state}{r+k'-1}$ for every $i$ and vector $\vec e_i $ consisting of elements from $\vec d$. Thus,
$\pi(\sem{f(\vec s)}{(\state,\vec d)}) = \sem{f(\vec s)}{(\calT,\pi(\vec d))}$ because $A$ and $B$ are $(r+k)$-similar and $k' \leq k$. The other cases are analogous. This concludes the proof of \mbox{Claim (A)}.

  Claim (B) can be proved by an induction over the nesting depth of $t$. The induction step uses Claim (A ii).

  For Claim (C) we have to prove that $\pi$ is witnessing the $r$-similarity of $A$ and $B$ in $\calS'$ and $\calT'$. The first property of similarity is trivial and the second follows from Claim (B). For the third property let $\vec c$ be an arbitrary $m$-tuple over $\nb{\calS'}{A}{r}$ and $R$ some $m$-ary relation symbol. Then $\vec c = (\sem{t_1}{(\calS',\vec c_1)},\ldots,\sem{t_n}{(\calS',\vec c_n)})$ with $\vec c_i$ over $A$ and $t_i \in \Terms{\auxSchema}{r}$. Thus $\vec c$ is a tuple over $\nb{\state}{A}{r\cdot k}$, by Claim (B), and therefore   $R^{\calS'}(\vec c)$ if and only if $R^{\calT'}(\pi(\vec c))$, by Claim (A).
\end{proof}

We now prove that unary \DynQF cannot maintain $s$-$t$-reachability. Intuitively, unary functions cannot store the transitive closure relation of a directed path in such a way, that the information can be extracted by a quantifier-free formula. The proof is simplified by the following observation.

\begin{lemma}\label{lemma:funct}
   If an $l$-ary query $\query$ can be maintained by a \DynQF-program, then $\query$ can be maintained by a $k$-ary \DynQF-program with only one $l$-ary auxiliary relation (used for storing the query result) on databases with at least two elements.
\end{lemma}
The restriction to structures with at least two elements is harmless, as we only use this lemma in a context where structures indeed have at least two elements.\\ \vspace{-2mm}

\begin{proofsketch}
  In order to encode relations by functions, two constants (i.e., $0$-ary functions) $c_\bot$ and $c_\top$ are used. Those constants are initialized by two distinct elements of the domain. Then a $k$-ary relation $R$ can be easily encoded by a $k$-ary function $f_R$ via $(a_1, \ldots, a_{k}) \in R$ if and only if \mbox{$f_R(a_1, \ldots, a_{k})=c_\top$}. \end{proofsketch}

\begin{theorem}\label{theorem:unaryfun}
  $\dynProb{\streachQ}$ is not in unary \DynQF. 
\end{theorem}  
\begin{proof}
   Towards a contradiction, we assume that  $\prog = (P, \init, Q)$ is a unary \DynQF-program that maintains $s$-$t$-reachability over schema $\schema = \inpSchema \cup \auxSchema$ with unary $\auxSchema$. By Lemma \ref{lemma:funct} we can assume that $\auxSchema$ contains only $0$-ary and unary function symbols and one 0-ary relation symbol $Q$ for storing the query result. The graphs used in this proof do not have self-loops and every node has at most one outgoing edge. Therefore we can assume, in order to simplify the presentation, that $\auxSchema$ contains a unary function symbol $e$, such that in every state $\state$ the function $e^\state$ encodes the edge relation $E$ as follows. If the single outgoing edge from $u$ is $(u,v)$ then  $e(u) = v$ and if $u$ has no outgoing edge then $e(u) = u$. 

Let $k$ be the nesting depth of $\prog$ and   let $n$ be chosen sufficiently large with respect to $\schema$ and $k$. Let $G = (V, E)$ be a graph where $V = \{s,t \} \cup A$ with $A = \{a_1, \ldots, a_n\}$ and
 $E = \{(a_i, a_{i+1}) \mid i \in \{1, \ldots, n-1\}\}$, i.e., $\restrict{G}{A}$ is a path of length $n-1$ from $a_1$ to $a_{n}$. Further, let $\state = (V, E, \aux)$ be the state obtained by applying $\init$ to $G$. 

   Our goal is to find $i$ and $j$ with $i<j$ such that for the two nodes $a \df a_i$ and $b\df a_j$ it holds $(a,b,s,t)\approx_{m}(b,a,s,t)$, where $m$ is the number from the substructure lemma for auxiliary functions  (Lemma \ref{lemma:substruclemmafun}), for modification sequences of length 2 and nesting depth $k$.

 Then, by Lemma \ref{lemma:substruclemmafun}, the program $\prog$ computes the same query result for the following two modification sequences:
    \begin{itemize}
      \item[($\beta_1$)] Insert edges $(s,a)$ and $(b, t)$. 
      \item[($\beta_2$)] Insert edges $(s,b)$ and $(a, t)$. 
    \end{itemize}

  However, applying the modification sequence $\beta_1$ yields a graph in which $t$ is reachable from $s$, whereas $\beta_2$ yields a graph in which $t$ is not reachable from $s$ (see Figure \ref{figure:theorem:unaryfun} for an illustration). This is the desired contradiction.

    \begin{figure}[!t]
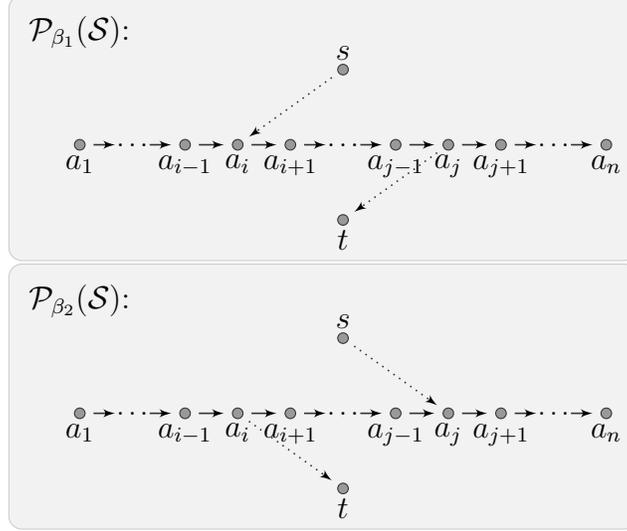
 
  \begin{center}
      \picunaryqfa
      \picunaryqfb
      \caption{The structure $\state$ from the proof of Theorem \ref{theorem:unaryfun}. Edges inserted by modification sequence $\beta_1$ and modification sequence $\beta_2$, respectively, are dotted. \label{figure:theorem:unaryfun}}
  \end{center}
    \end{figure}
  Thus it remains to show the existence of such $i$ and $j$. To this end, let $t_1, \ldots, t_l$ be the lexicographic enumeration of $\Terms{\schema}{k}$ with respect to some fixed order of the function symbols.   Let the \textit{$k$-neighborhood vector} $\nbv{c}{\state}{k}$ of an
  element $c$ in $\calS$ be the tuple
  $(c, t_1(c), \ldots, t_l(c))$. For a tuple $\vec
  c = (c_1, \ldots, c_m)$, the $k$-neighborhood vector $\nbv{\vec
   c}{\state}{k}$ of $\vec c$ is the tuple $(\nbv{c_1}{\state}{k},
  \ldots, \nbv{c_m}{\state}{k})$.  The number of equality types of such neighborhood vectors is finite and bounded by a number that only depends on $m$, $k$ and $\auxSchema$.

By applying Ramsey's theorem on the graph over $\{1,\ldots,n\}$, where each pair $(i,j)$ with $i<j$ is colored by the equality type of  $\nbv{a_i, a_j, s, t}{\state}{m+1}$, we obtain numbers $i_1<i_2<i_3$ such that  the equality types of $\nbv{a_{i_1},a_{i_2},s,t}{\state}{m+1}$, $\nbv{a_{i_1},a_{i_3},s,t}{\state}{m+1}$, and $\nbv{a_{i_2},a_{i_3},s,t}{\state}{m+1}$ are equal. In particular, as all function symbols are unary, the equality types of $\nbv{a_{i_1},s,t}{\state}{m+1}$, and $\nbv{a_{i_2},s,t}{\state}{m+1}$ and finally those of  $\nbv{a_{i_1},a_{i_2},s,t}{\state}{m+1}$ and $\nbv{a_{i_2},a_{i_1},s,t}{\state}{m+1}$ are equal. 

For the latter conclusion, we show the following claim: if  for two terms $t_1$ and $t_2$ of depth at most $m+1$ it holds $t_1(a_{i_1})=t_2(a_{i_2})$ then also $t_1(a_{i_2})=t_2(a_{i_1})$. We observe that if $t_1(a_{i_1})=t_2(a_{i_2})$  then also $t_1(a_{i_1})=t_2(a_{i_3})$ and $t_1(a_{i_2})=t_2(a_{i_3})$ (since $\nbv{a_{i_1},a_{i_2},s,t}{\state}{m+1}$, $\nbv{a_{i_1},a_{i_3},s,t}{\state}{m+1}$, and $\nbv{a_{i_2},a_{i_3},s,t}{\state}{m+1}$ have the same equality type). Hence,   $t_1(a_{i_2})=t_2(a_{i_2})$ and therefore $t_1(a_{i_2})=t_2(a_{i_2})=t_1(a_{i_1})=t_2(a_{i_1})$. The latter equality follows as the equality types of $\nbv{a_{i_1},s,t}{\state}{m+1}$, and $\nbv{a_{i_2},s,t}{\state}{m+1}$ are equal. This concludes the proof of the claim.

  To prove $(a,b,s,t)\approx_{m}(b,a,s,t)$ it only remains to show that $(u,v) \in E$ if and only if $(u', v')\in E$, for two components $u$ and $v$ from $\nbv{a,b,s,t}{\state}{m}$ and their corresponding components $u'$ and $v'$ from $\nbv{b,a,s,t}{\state}{m}$. However, $(u,v)\in E$ if and only if $e(u) = v$, and analogously $(u', v')\in E$ if and only if $e(u') = v'$. Thus this claim follows already from the fact that $\nbv{a_{i_1},a_{i_2},s,t}{\state}{m+1}$ and $\nbv{a_{i_2},a_{i_1},s,t}{\state}{m+1}$ have the same equality type.
\end{proof}

We now extend the lower bound for invariant initialization obtained in Theorem \ref{theorem:logicinit} to
quantifier-free programs with auxiliary functions. Invariant initialization is still weak in the presence of auxiliary functions in
the sense, that functions initialized by invariant initialization can only point to 'distinguished' nodes, as formalized by the following lemma.
\begin{lemma}\label{lemma:initfunc}
  Let $\prog = (P, \init, Q)$ be a \DynQF-program with invariant initialization mapping $\init$ and auxiliary schema $\auxSchema$. Further let $\inp$ be an input structure for $\prog$ whose domain contains $b$ and $b'$ with \mbox{$b \neq b'$}. If $\swap{b}{b'}$ is an isomorphism of $\inp$, then \mbox{$f^{\init(\inp)}(\vec a) \neq b$} for all $k$-ary function symbols $f \in \auxSchema$ and all $k$-tuples $\vec a$.
\end{lemma}
\begin{proof}
  The claim follows immediately from the invariance of the initialization mapping.
\end{proof}

The following lemma will be useful for the proof of the next theorem.

\begin{lemma} \label{lemma:dynqfiso}
  Let $\prog$ be a \DynQF program and $\state$ and $\calT$ be states of $\prog$ with domains $S$ and $T$. Further let $A \subseteq S$ and $B \subseteq T$ be closed. If $\restrict{\state}{A}$ and $\restrict{\calT}{B}$ are isomorphic via $\pi$ then $\restrict{\updateState{P}{\alpha}{\state}}{A}$ and $ \restrict{\updateState{P}{\beta}{\calT}}{B}$ are isomorphic via $\pi$ for all $\pi$-respecting modification sequences $\alpha$, $\beta$  on $A$ and $B$.
\end{lemma}
\begin{proof}
  Observe that when $A$ and $B$ are closed and $\restrict{\state}{A}$ and $\restrict{\calT}{B}$ are isomorphic via $\pi$ then $A$ and $B$ are $k$-similar via $\pi$ for arbitrary $k$. Thus the claim follows from Lemma \ref{lemma:substruclemmafun}.
\end{proof}

\begin{theorem}\label{theorem:logicinitqf}
  $\dynProb{\streachQ}$ cannot be maintained in \DynQF with invariant initialization mapping. This holds even for 1-layered \stgraphs.
\end{theorem}
\begin{proof}
  We follow the argumentation of the proof of Theorem \ref{theorem:logicinit}.

  Towards a contradiction, assume that $\prog$ is a $\DynQF$-program with auxiliary schema $\auxSchema$ and invariant initialization mapping $\init$ which maintains the $s$-$t$-reachability query for 1-layered \stgraphs. Let $m$ be the maximum arity of relation  or function symbols in  $\auxSchema \cup \{E\}$. Further let $n$ be the number of isomorphism types of structures with at most $m+2$ elements. 

  We consider the complete 1-layered \stgraphs $G_i = (V_i, E_i)$, $2 \leq i \leq n+2$,  with $V_i = \{s,t\} \cup A_i$ and $A_i = \{a_1, \ldots, a_{i}\}$. Further let $\state_i = (V_i, E_i, \aux_i)$ be the state obtained by applying $\init$ to $G_i$.

  We observe that $\swap{a}{a'}$ is an automorphism of $G_i$ for all pairs $(a,a')$ of nodes in $A_i$ with $a \neq a'$. Thus, by Lemma \ref{lemma:initfunc}, $s$ and $t$ are the only values that the auxiliary functions in $\state_i$ can assume, and therefore $\restrict{\state_i}{A \cup \{s,t\}}$ is closed for any subset $A$ of $A_i$. Hence,  by Lemma \ref{lemma:dynqfiso}, it is sufficient
 to find $\state_k$ and $\state_l$ with $k<l$ such that $\state_k$ is isomorphic to $\restrict{\state_l}{V_k}$. Then, we can apply the same sequences of modifications as in Theorem \ref{theorem:logicinit} to reach a contradiction.

  Recall that a tuple is diverse, if all components differ pairwise. Since \linebreak[4] \mbox{$G_i \isomorphVia{\swap{\vec a}{\vec b}} G_i$}, for two diverse $m'$-tuples $\vec a$ and $\vec b$ over $A_i$ with $m' \leq m$, also $\state_i \isomorphVia{\swap{\vec a}{\vec b}} \state_i$ by the invariance of $\init$. In particular $(s,t,\vec a)$ and $(s,t,\vec b)$ are of the same isomorphism type.

  Since $n$ is the number of isomorphism types of structures of at most $m+2$ elements, there are two states $\state_k$ and $\state_l$ such that, all diverse $m$-tuples over $A_k$ and $A_l$ extended by $s$ and $t$ are of the same isomorphism type in $\state_k$ and $\state_l$, respectively. But then  $\state_k \isomorph \restrict{\state_l}{V_k}$.
\end{proof}

  \section{Lower Bounds for Other Dynamic Queries}\label{section:moreproblems}
    \makeatletter{}
In this section we use the lower bounds obtained for the dynamic $s$-$t$-reachability query for shallow graphs to establish lower bounds for the dynamic variants of the following Boolean queries

\problemdescr{\clique{k}}{A graph $G$}{Does $G$ contain a $k$-clique?}

\problemdescr{\colorability{k}}{A graph $G$}{Is $G$ $k$-colorable?}

\noindent where $k$ is a fixed natural number. Cliques are usually defined for undirected graphs only. We define a
clique in a directed graph to be a set of nodes such that each pair of nodes from the set is connected by an edge. Similarly for colorability.

Lower bounds for the dynamic variants of the $\clique{k}$ and $\colorability{k}$ problems (where $k$ is fixed) can be established via reductions to the dynamic $s$-$t$-reachability query for shallow graphs.

\begin{proposition}\label{proposition:binary}
  The dynamic query \dynClique{k}, for $k \geq 3$, and the dynamic query \dynColorability{k}, for $k \geq 2$,  are not in binary \DynPropbi.
\end{proposition}
\begin{proof}
   We prove that \dynClique{3} cannot be maintained in binary \DynProp. Afterwards we sketch the proof for \dynClique{k}, for arbitrary $k \geq 3$. The graphs used in the proof have a $k$-Clique if and only if they are not  $(k-1)$-colorable. Therefore it follows that \dynColorability{k} cannot be maintained in binary \DynProp.

  More precisely, we show that from a binary \DynProp-program $\prog'$ for the query \mbox{\dynClique{3}}  one can construct a dynamic program $\prog$ that maintains the $s$-$t$-reacha\-bility query for 2-layered \stgraphs. As the latter does not exist thanks to Theorem \ref{theorem:binary}, we can conclude that the former does not exist either.

  Let us thus assume that $\prog'=(P', \init', Q')$ is a dynamic program for \mbox{\dynClique{3}} with binary auxiliary schema $\auxSchema'$ and built-in schema $\builtinSchema'$.

  The reduction is very simple. For a $2$-layered graph $G=(\{s,t\}\cup A\cup B,E)$, let $G'$ be the graph obtained from $G$ by identifying $s$ and $t$. Clearly, $G$ has a path from $s$ to $t$ if and only if $G'$ has a $3$-clique. See Figure \ref{figure:reachtoclique} for an illustration.

  \begin{figure}[t]
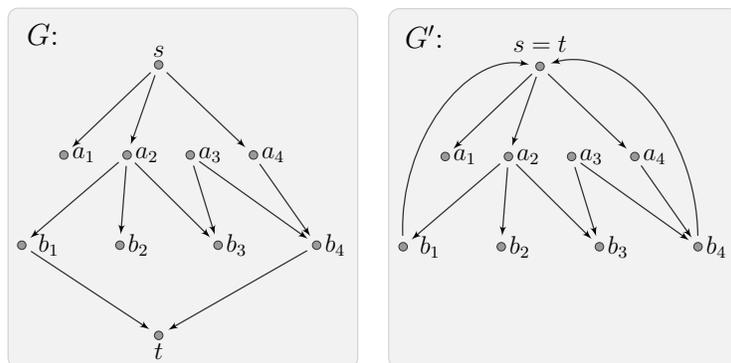
 
      \begin{center}
      \scalebox{0.8}{
        \piccliquea
      }
      \scalebox{0.8}{
        \piccliqueb
      }
    \end{center}
      \caption{The construction from Proposition \ref{proposition:binary}. The $s$-$t$-paths $(s, a_2, b_1, t)$ and $(s, a_4, b_4, t)$ in $G$ correspond to the cliques $\{s, a_2, b_1\}$ and $\{s, a_4, b_4\}$ in $G'$. \label{figure:reachtoclique}}
  \end{figure}

  The dynamic program $\prog$ uses the same auxiliary schema as $\prog'$, the same initialization mapping and the same built-in schema relations. However, edges $(u,t)$ in $E$ are interpreted as if they were edges $(u,s)$ in $E'$. More precisely, the update formulas of $\prog$ are obtained from those in $\prog'$ by replacing every atomic formula $E'(x,y)$ with $(y=s\land E(x,t))\lor (y\not=s \land E(x,y))$. Obviously, $\prog$ is a dynamic program for $s$-$t$-reachability for 2-layered \stgraphs if $\prog'$ is a dynamic program for 
\dynClique{3}, as desired.

  For arbitrary $k$, the construction is similar. The idea is that $\prog$ simulates on a graph $G$ the behavior of $\prog'$ on $G\otimes K_{k-3}$, that is, the graph that results from $G$ by adding a $(k-3)$-clique and completely connecting it with every node of $G$. Interestingly, the update formulas of $\prog$ are exactly as in the previous reduction to  \dynClique{3}, as the ``virtual'' additional $k-3$ nodes are never involved in changes of the graph. However, $\init$ is not the same as $\init'(G)$ but rather the projection of $\init'(G\otimes K_{k-3})$ to the nodes of $G$.
\end{proof}
\medskip

\begin{proposition}\label{proposition:invariant}
  The dynamic query \dynClique{k}, for $k \geq 3$, and the dynamic query \dynColorability{k}, for $k \geq 2$, cannot be maintained in \DynQF with invariant initialization mapping.
\end{proposition}
\begin{proof}
  The proof approach is the same as for the previous proposition. We prove that \dynClique{3} cannot be maintained in \DynQF with invariant initialization. Afterwards we sketch the proof for  \dynClique{k}, for arbitrary $k \geq 3$. The graphs used in the proof have a $k$-Clique if and only if they are not  $(k-1)$-colorable. Therefore it follows that \dynColorability{k}  cannot be maintained in \DynQF with invariant initialization mapping.

  More precisely, we show that from \DynQF dynamic program $\prog'$ with invariant initialization that maintains \dynClique{3}  one can construct a dynamic program $\prog'$ that maintains the $s$-$t$-reachability query for 1-layered \stgraphs. As the latter does not exist thanks to Theorem \ref{theorem:logicinitqf}, we can conclude that the former does not exist either.

  Let us thus assume that $\prog'=(P', \init', Q')$ is a dynamic program for \mbox{\dynClique{3}} with invariant initialization mapping \init' and auxiliary \linebreak[4] \mbox{schema $\auxSchema'$}.

  We use the following simple reduction. For a $1$-layered graph $G=(\{s,t\}\cup A,E)$, let $G'$ be the graph obtained from $G$ by adding an edge $(s,t)$. Clearly, $G$ has a path from $s$ to $t$ if and only if $G'$ has a $3$-clique. 

  The dynamic program $\prog$ uses the same auxiliary schema as $\prog'$ and the same initialization mapping. The update formulas of $\prog$ are obtained from those in $\prog'$ by replacing every atomic formula $E'(x,y)$ with $(E(x, y) \vee (x=s \wedge y = t))$. Obviously, $\prog$ is a dynamic program for $s$-$t$-reachability for 2-layered \stgraphs if $\prog'$ is a dynamic program for \dynClique{3}, as desired.

  For arbitrary $k$, the construction is similar. The idea is that $\prog$ simulates on a graph $G$ the behavior of $\prog'$ on $G\otimes (K_{k-3}, K_{k-3})$, that is, the graph that results from $G$ by adding two $(k-3)$-cliques and completely connecting them with every node of $G$. The update formulas of $\prog$ are exactly as in the previous reduction to  \dynClique{3}. However, $\init$ is not the same as $\init'(G)$ but rather the projection of $\init'(G\otimes (K_{k-3}, K_{k-3}))$ to the nodes of $G$. By \mbox{Lemma \ref{lemma:initfunc}}, auxiliary functions in $\init(G)$ do not take values from $(K_{k-3}, K_{k-3})$. Thus $\prog$ is a dynamic program for $s$-$t$-reachability for 2-layered \stgraphs if $\prog'$ is a dynamic program for \dynClique{k}.  
\end{proof}

  \section{Subclasses of \DynProp}\label{section:normalforms}
    \makeatletter{}Proving that Reachability cannot be maintained in $\DynProp$ appears to be non-trivial. A natural question is, whether lower bounds for syntactic fragments of $\DynProp$ can be proved, without restrictions on the arity of auxiliary relations. Normal form results from \cite{ZeumeS13reachmfcs} (see below) imply that lower bounds for some large fragments cannot be obtained easier than for \DynProp. In this section we prove that Reachability cannot be maintained in the (very) weak fragment of \DynProp where update formulas are restricted to be conjunctions of atoms.

Throughout this section we allow arbitrary initialization and no auxiliary functions.

A formula is \textit{negation-free} if it does not use negation at all. A formula is \textit{conjunctive} if it is a conjunction of (positive or negated) literals. A dynamic program is negation-free (conjunctive, respectively) if all its update formulas are negation-free  (conjunctive, respectively). We follow the naming schema from \cite{ZeumeS14CQicdt} and refer to the conjunctive, the  negation-free and the conjunctive, negation-free fragment of $\DynProp$ as $\DynPropCQneg$, $\DynPropUCQ$ and $\DynPropCQ$, respectively.

The following theorem implies that lower bounds for \DynPropCQneg and \DynPropUCQ immediately yield lower bounds for \DynProp. In other words, proving lower bounds for those fragments is not easier than proving lower bounds for $\DynProp$ itself.

\begin{theorem}[\cite{ZeumeS13reachmfcs, ZeumeS13reacharxiv}] \label{theorem:normalforms}
  Let $\query$ be a query. Then the following statements are equivalent:
  \begin{enumerate}
    \item $\query$ can be maintained in $\DynProp$.
    \item $\query$ can be maintained in $\DynPropCQneg$.
    \item $\query$ can be maintained in $\DynPropUCQ$.
  \end{enumerate}
\end{theorem}

The rest of this section is devoted to the proof of the following theorem. 

\begin{theorem}\label{theorem:reachnotindynpropcq}
    $\dynProb{\streachQ}$ is not in \DynPropCQ.
\end{theorem}

To this end, we first  prove that the query \problem{NonEmptySet} from Example \ref{example:emptylist} cannot be maintained in this fragment. Afterwards we sketch how to adapt this proof for the reachability query.

 For technical reasons, the proof assumes a $\DynAND$-program in which no atom contains any variable more than once. We first illustrate by an example how this restriction can be achieved.
  \begin{example}\label{example:dynandnormalform}
    We consider the following $\DynAND$-program, where, for simplicity, only update formulas for insertions are specified.
      \begin{align*}
        \uf{R}{\ins}{u}{x,y} &= S(x,y) \land R(x,x) \\
        \uf{S}{\ins}{u}{x,y} &= S(x,y)
      \end{align*}
  An equivalent $\DynAnd$-program in which all update formulas only contain atoms with distinct variables can be obtained by replacing $R(x,x)$ by $R'(x)$ where $R'$ is a fresh unary relation symbol. It then has to be ensured, that $R'(x)\equiv R(x,x)$. This can be achieved by updating $R'$ with the update formula for $R$, in which $x$ and $y$ are unified.
      \begin{align*}
        \uf{R}{\ins}{u}{x,y} &= S(x,y) \land R'(x) \\
        \uf{S}{\ins}{u}{x,y} &= S(x,y) \\
        \uf{R'}{\ins}{u}{x} &= S(x,x) \land R'(x)
      \end{align*}

     Finally we apply the same construction to the atom $S(x,x)$ in $ \ufwa{R'}{\ins}$:
      \begin{align*}
        \uf{R}{\ins}{u}{x,y} &= S(x,y) \land R'(x) \\
        \uf{S}{\ins}{u}{x,y} &= S(x,y) \\
        \uf{R'}{\ins}{u}{x} &= S'(x) \land R'(x) \\
        \uf{S'}{\ins}{u}{x} &= S'(x)
      \end{align*}
      
  \end{example}

  The process of Example \ref{example:dynandnormalform} necessarily terminates since there is only a finite number of equality types for the variables of each of the atoms occurring in an update formula. An \emph{equality type} $\rho$ of a set of variables $X = \{x_1, \ldots, x_n\}$ is simply an equivalence relation on $X$.   
  \begin{lemma}\label{lemma:dynandnormalform}
    For every \DynAND-program there is an equivalent \DynAND-program in which no atom in any update formula contains a variable more than once.
  \end{lemma}
  \begin{proofsketch}
    For a given $\DynAND$-program $\prog$ schema $\schema$, construct an equivalent $\DynAnd$-program $\prog'$ over schema $\schema'$ where $\schema'$ contains, for every $k$-ary relation symbol $R \in \schema$ and every equality type $\rho$ on $k$ variables $x_1, \ldots, x_k$, a relation symbol $R^\rho$ of arity $k'$ where $k'$ is the number of equivalence classes of $\rho$. 

      The intention is that $(\state,\beta)\models R(\vec x)$, for  a state $R$ and variable assignment $\beta$ respecting $\rho$ if and only if $(\state,\beta^\rho) \models R^\rho(\vec y)$, where $\beta^\rho$ maps every variable $y_i$ to the value of the $i$-th equivalence class of $\rho$ under $\beta$. This can be ensured along the lines of Example \ref{lemma:dynandnormalform}.
  \end{proofsketch}

  We prove Theorem \ref{theorem:separationdynanddynprop} in a slightly more general setting. A modification $\alpha$ is \emph{honest} with respect to a given state if it does not insert a tuple already present in the input database and does not delete a tuple which is not present in the database. A query is in  h-$\DynC$ if it can be maintained with $\calC$  update programs, for all sequences of honest modifications. It is easy to see  that for a class $\calC$ closed under boolean operations, the classes $\DynC$ and h-$\DynC$ coincide. However for weak classes such as $\DynAND$ the restriction to honest modifications might make a difference, since update formulas cannot explicitly test (at least not in a straight forward way) whether a modification is honest. Nevertheless, all our proofs work for both kinds of types of modifications.
  
  We prove that h-$\DynAnd$ (and therefore also $\DynAnd$) cannot maintain the query $\exists x U(x)$ from Example \ref{example:emptylist}.   

  \begin{lemma}\label{theorem:separationdynanddynprop}
    $\dynProb{NonEmptySet}$ is neither in \DynAND nor in h-$\DynAnd$.
  \end{lemma}
  \begin{proof}
    Towards a contradiction, we assume that there is a h-$\DynAnd$-program $\prog = (P, \init, Q)$ over schema $\schema$ that maintains query $\query$ defined by $\exists x U(x)$ and, by Lemma \ref{lemma:dynandnormalform}, that no variable occurs more than once in any atom of an update formula of $\prog$. 

    The following notions will be convenient for the proof. The \emph{dependency graph} of a dynamic program $\prog$ with auxiliary schema $\schema$ has vertex set $V = \schema$ and an edge $(R, R')$ if the relation symbol $R'$ occurs in one of the update formulas for $R$. The \emph{deletion dependency graph} of $\prog$ is defined like the dependency graph except that only update formulas for deletions are used. The \emph{deletion depth} of a relation $R$ is defined as the length of the shortest path from $Q$ to $R$ in the deletion dependency graph. 

We start with a simple observation. Let $R(u)$ be a relation atom in the formula $\uf{Q}{\del}{u}{}$ for the 0-ary query relation $Q$, that is:
$$\uf{Q}{\del}{u}{} \df \ldots \wedge R(u) \wedge \ldots$$
Further let $\state$ be a state in which the relation $U$ contains two elements $a\not=b$. Then, necessarily, $R^\state$ contains $a$ and $b$, as otherwise deletion of $a$ or $b$ would make $Q$ empty without $U$ becoming empty. This observation can be generalized: if a relation $R$ has ``distance $k$'' from $Q$ in the subgraph of the dependency graph induced by $\del$-formulas and $U$ contains at least $k+1$ elements, then $R$ must contain all \emph{diverse} tuples over $U$, that is, tuples that consist of pairwise distinct elements from $U$.

We prove this observation next, afterwards we look at how the statement of the lemma follows. Using our assumption on non-repeating variables, it is easy to show that the arity of relations of deletion depth $k$ is at most $k$ (at most one plus the arity of the updated relation). 

We prove by induction on $k$ that, for each relation $R$ of deletion depth $k$, and every state $\state$ in which $U$ contains at least $k+1$ elements, $R$ has to contain all diverse tuples over $U$. 

For $k=0$ this is obvious as $Q$ needs to contain the empty tuple if $U$ is non-empty.

 For $k>0$, let $\state$ be a state such that $U^\state$ contains at least $k+1$ elements. Further let $R$ be some arbitrary relation symbol of deletion distance $k$. Then $R(\vec x)$  occurs in the update formula $\uf{R'}{\del}{u}{\vec y}$ of some relation symbol $R'$ of deletion depth $k-1$ for some $\vec x = (x_1, \ldots, x_l)$,  with $\vec x \subseteq \{u\} \cup \vec y$. By the above, $l\leq k$ and $\vec y$ contains at most $k-1$ variables. 

Towards a contradiction, let us assume that there is a diverse $k$-tuple $\vec a = (a_1, \ldots, a_k)$ over $U^\state$ that is not in $R^\state$. Let $\Theta: \{x_1, \ldots, x_l\} \rightarrow U^\state$ be the assignment with $\Theta(x_i) = a_i$ and let $\hat{\Theta}$ be some extension of $\Theta$ to an injective assignment of $\{u\} \cup \vec y$ to elements from $U^\state$ (such an assignment exists because $|\{u\} \cup \vec y| \leq k< |U|$). Then $\uf{R'}{\del}{u}{\vec y}$ evaluates to false in state $\state$ under $\hat{\Theta}$ (since $\vec a \notin R^\state$ by assumption). Thus, deleting $\hat{\Theta}(u)$ from $U^\state$ yields a state $\state'$ with $\hat{\Theta}(\vec y) \notin R'^{\state'}$. However, $U^{\state'}$ still contains at least $k$ elements and therefore, by induction hypothesis, the relation $R'^{\state'}$ contains every diverse tuple over $U^{\state'}$ and thus, in particular, $\hat{\Theta}(\vec y)$, the desired contradiction from the assumption that $\vec a\not\in R^\state$.

Now we can complete the proof of Lemma \ref{theorem:separationdynanddynprop}. Let $\state$ be a state in which the set $U$ contains $m+1$ elements, where $m$ is the maximum (finite) deletion depth of any relation symbol in $\prog$. By the claim above, all relations whose symbols are reachable from $Q$ in the deletion dependency graph of $\prog$ contain all diverse tuples over $U^\state$. Thus, all relation atoms over tuples from $U^\state$ evaluate to true. It is easy to show by induction on the length of modification sequences  that this property (applied to $U^{\state'}$)  holds for all states $\state'$ that can be obtained from $\state$ by deleting elements from $U^\state$. In particular, it holds for any such state in which $U^{\state'}$ contains only one element $a$. But then, $\uf{Q}{\del}{a}{}$ evaluates to true in $\state'$ and thus $Q$ remains true after deletion of $a$, the desired contradiction to the assumed correctness of $\prog$.
  \end{proof}

\begin{proofsketchof}{Theorem \ref{theorem:reachnotindynpropcq}}
  Towards a contradiction assume that there is a \DynAND-program $\prog$ for $\dynProb{\streachQ}$ over schema $\schema$. We show that a \DynAND-program $\prog'$ can be constructed from $\prog$ such that $\prog'$ maintains $\dynProb{NonEmptySet}$ under deletions. As the proof of the preceding lemma shows that $\dynProb{NonEmptySet}$ cannot be maintained in \DynAND even if elements are deleted from $U$ only, this is the desired contradiction. 
  
  The intuition behind the construction of $\prog'$ is as follows. For sets $U\subseteq A$, the 1-layered graph $G$ with nodes $\{s, t\} \cup A$ and edges $\{(s, a) \mid a \in U\} \cup \{(a, t) \mid a \in A\}$ naturally corresponds to the instance $I$ of $\dynProb{NonEmptySet}$ over domain $A$ with set $U$. The deletion of an element $a$ from $U$ in $I$ corresponds to the deletion of the edge $(s, a)$ from $G$. Using this correspondence, the program $\prog'$ essentially maintains the same auxiliary relations as $\prog$. When $a$ is deleted from $U$ then $\prog'$ simulates $\prog$ after the deletion of $(s, a)$. 
  
  A complication arises from the fact that $\dynProb{NonEmptySet}$ does not have constants $s$ and $t$. Therefore the program $\prog'$ encodes the relationship of $s$ and $t$ to elements from $A$ by using additional auxiliary relations. More precisely, for every $k$-ary relation symbol $R \in \schema$ and every tuple $\rho = (\rho_1, \ldots, \rho_k)$ over $\{\bullet, s, t\}$, the program $\prog'$ has a fresh $l$-ary relation symbol $R^\rho$ where $l$ is the number of $\rho_i$'s with $\rho_i = \bullet$. The intention is as follows. Let $i_1 < \ldots < i_l$ such that $\rho_{i_j} = \bullet$. With every $l$-tuple $\vec u=(y_1,\ldots,y_l)$ of variables we associate the tuple ${\vec u}^\rho=(u_1^\rho,\ldots,u_k^\rho)$ of terms from $\{s,t,y_1,\ldots,y_l\}$, where (1) $u_i^\rho = s$ if $\rho_i = s$, (2) $u_i^\rho = t$ if $\rho_i = t$, and (3) $u_{i_j}^\rho=y_j$, for $j \in \{1,\ldots,l\}$. Analogously, we define ${\vec a}^\rho$ for tuples $\vec a=(a_1,\ldots,a_l)$ over $A$. Then $\prog'$ 
ensures 
that  $\vec a \in R^\rho$ in some state if and only if ${\vec a}^\rho \in R$ in the corresponding state of $\prog$.   
  
  Update formulas $\uf{R^\rho}{\del U}{v}{x_1, \ldots, x_l}$ of $\prog'$ are obtained from update formulas $\uf{R}{\del E}{u, v}{x_1, \ldots, x_k}$ of $\prog$ in two steps. First, from $\ufwa{R}{\del E}$ a formula $\phi'$ is constructed by replacing every occurrence of $x_i$ by $x^\rho_i$ and replacing every occurrence of $u$ by $s$. Then  $\ufwa{R^\rho}{\del U}$ is obtained from $\phi'$ by replacing every atom $T(\vec w)$ in $\ufwa{R}{\del E}$ by $T^{\rho}(\vec y)$, for the unique tuple $\vec y$ of variables and the unique tuple $\rho$, for which ${\vec y}^\rho=\vec w$.
  
  Now, $\prog'$ yields the same query result after deletion of elements $a_1, \ldots, a_m$ as $\prog$ after deletion of edges $(s,a_1), \ldots, (s, a_m)$. Hence the program $\prog'$ maintains $\dynProb{NonEmptySet}$ under deletions. This is a contradiction. 
\end{proofsketchof}


  \section{Future Work}
    The question whether Reachability is maintainable with first-order updates remains one of the major open questions in dynamic complexity.  Proving that
    Reachability cannot be maintained with quantifier-free updates
    with arbitrary auxiliary data seems to be a worthwhile
    intermediate goal, but it appears non-trivial as well.

    We contributed to the intermediate goal by giving a first lower bound for  binary auxiliary
    relations. Whether the strictness of the  arity hierarchy for \DynProp extends beyond arity three 
   is another open question. 

    For (full) first-order updates a major challenge is the development of lower bound tools. Current techniques are in some sense not fully dynamic: either results from static descriptive complexity are applied to constant-length modification sequences; or non-constant but very regular modification sequences are used.
    In the latter case, the modifications do not depend on previous changes
    to the auxiliary data (as, e.g., in \cite{GraedelS12} and
    in this paper). Finding techniques that adapt to changes could be a good
    starting point.

  \bibliographystyle{plain}
  \bibliography{bibliography}

\begin{thebibliography}{10}

\bibitem{DongLW03}
Guozhu Dong, Leonid Libkin, and Limsoon Wong.
\newblock Incremental recomputation in local languages.
\newblock {\em Inf. Comput.}, 181(2):88--98, 2003.

\bibitem{DongS98}
Guozhu Dong and Jianwen Su.
\newblock Arity bounds in first-order incremental evaluation and definition of
  polynomial time database queries.
\newblock {\em J. Comput. Syst. Sci.}, 57(3):289--308, 1998.

\bibitem{DongZ00}
Guozhu Dong and Louxin Zhang.
\newblock Separating auxiliary arity hierarchy of first-order incremental
  evaluation systems using (3k+1)-ary input relations.
\newblock {\em Int. J. Found. Comput. Sci.}, 11(4):573--578, 2000.

\bibitem{Etessami98}
Kousha Etessami.
\newblock Dynamic tree isomorphism via first-order updates.
\newblock In Alberto~O. Mendelzon and Jan Paredaens, editors, {\em Proceedings
  of the Seventeenth ACM SIGACT-SIGMOD-SIGART Symposium on Principles of
  Database Systems, June 1-3, 1998, Seattle, Washington, USA}, pages 235--243.
  ACM Press, 1998.

\bibitem{GeladeMS12}
Wouter Gelade, Marcel Marquardt, and Thomas Schwentick.
\newblock The dynamic complexity of formal languages.
\newblock {\em ACM Trans. Comput. Log.}, 13(3):19, 2012.

\bibitem{GraedelS12}
Erich Gr{\"a}del and Sebastian Siebertz.
\newblock Dynamic definability.
\newblock In Alin Deutsch, editor, {\em 15th International Conference on
  Database Theory, ICDT '12, Berlin, Germany, March 26-29, 2012}, pages
  236--248. ACM, 2012.

\bibitem{GrahamRS1990}
R.L. Graham, B.L. Rothschild, and J.H. Spencer.
\newblock {\em Ramsey Theory}.
\newblock Wiley Series in Discrete Mathematics and Optimization. Wiley, 1990.

\bibitem{Hesse01}
William Hesse.
\newblock The dynamic complexity of transitive closure is in
  {DynTC}$^{\mbox{0}}$.
\newblock In {\em Database Theory - ICDT 2001, 8th International Conference,
  London, UK, January 4-6, 2001, Proceedings}, pages 234--247, 2001.

\bibitem{Hesse03}
William Hesse.
\newblock {\em Dynamic Computational Complexity}.
\newblock PhD thesis, University of Massachusetts Amherst, 2003.

\bibitem{PatnaikI94}
Sushant Patnaik and Neil Immerman.
\newblock {Dyn-FO}: A parallel, dynamic complexity class.
\newblock In {\em Proceedings of the Thirteenth ACM SIGACT-SIGMOD-SIGART
  Symposium on Principles of Database Systems, May 24-26, 1994, Minneapolis,
  Minnesota}, pages 210--221. ACM Press, 1994.

\bibitem{PatnaikI97}
Sushant Patnaik and Neil Immerman.
\newblock {Dyn-FO}: A parallel, dynamic complexity class.
\newblock {\em J. Comput. Syst. Sci.}, 55(2):199--209, 1997.

\bibitem{PatrascuD04}
Mihai Patrascu and Erik~D. Demaine.
\newblock Lower bounds for dynamic connectivity.
\newblock In L{\'a}szl{\'o} Babai, editor, {\em Proceedings of the 36th Annual
  ACM Symposium on Theory of Computing, Chicago, IL, USA, June 13-16, 2004},
  pages 546--553. ACM, 2004.

\bibitem{SchmitzS11}
Sylvain Schmitz and Ph. Schnoebelen.
\newblock Multiply-recursive upper bounds with {H}igman's lemma.
\newblock In {\em Automata, Languages and Programming - 38th International
  Colloquium, ICALP 2011, Zurich, Switzerland, July 4-8, 2011, Proceedings,
  Part II}, pages 441--452, 2011.

\bibitem{WeberS07}
Volker Weber and Thomas Schwentick.
\newblock Dynamic complexity theory revisited.
\newblock {\em Theory Comput. Syst.}, 40(4):355--377, 2007.

\bibitem{ZeumeS13reachmfcs}
Thomas Zeume and Thomas Schwentick.
\newblock On the quantifier-free dynamic complexity of reachability.
\newblock In Krishnendu Chatterjee and Jiri Sgall, editors, {\em MFCS}, volume
  8087 of {\em Lecture Notes in Computer Science}, pages 837--848. Springer,
  2013.

\bibitem{ZeumeS13reacharxiv}
Thomas Zeume and Thomas Schwentick.
\newblock On the quantifier-free dynamic complexity of reachability.
\newblock {\em CoRR}, abs/1306.3056, 2013.

\bibitem{ZeumeS14CQicdt}
Thomas Zeume and Thomas Schwentick.
\newblock Dynamic conjunctive queries.
\newblock In Nicole Schweikardt, Vassilis Christophides, and Vincent Leroy,
  editors, {\em ICDT}, pages 38--49. OpenProceedings.org, 2014.

\end{thebibliography}

\end{document}